\documentclass[11pt]{article}
\usepackage[T1]{fontenc}
\usepackage{amsfonts}
\usepackage{amsmath}
\usepackage{enumerate}
\usepackage{amssymb}
\usepackage{amsthm}
\usepackage{bbm}
\usepackage{bm}
\usepackage{mathrsfs}
\usepackage{verbatim}
\usepackage{setspace}
\usepackage{color}
\usepackage{pdfsync}
\usepackage{enumitem}
\theoremstyle{plain}

\newtheorem{theorem}{Theorem}[section]

\newtheorem{proposition}[theorem]{Proposition}
\newtheorem{corollary}[theorem]{Corollary}
\newtheorem{lemma}[theorem]{Lemma}
\newtheorem{remark}[theorem]{Remark}
\newtheorem{example}[theorem]{Example}
\newtheorem{examples}[theorem]{Examples}
\newtheorem{foo}[theorem]{Remarks}
\newenvironment{Example}{\begin{example}\rm}{\end{example}}

\newenvironment{Remark}{\begin{remark}\rm}{\end{remark}}

\newcommand{\be}{\begin{equation}}
\newcommand{\ee}{\end{equation}}
\newcommand{\bea}{\begin{eqnarray}}
\newcommand{\eea}{\end{eqnarray}}
\newcommand{\beas}{\begin{eqnarray*}}
  \newcommand{\eeas}{\end{eqnarray*}}

\newcommand{\brak}[1]{\ensuremath{\left( #1 \right)}}
\newcommand{\crl}[1]{\ensuremath{ \left\{ #1 \right\} }}
\newcommand{\edg}[1]{\ensuremath{ \left[ #1 \right] }}
\newcommand{\ang}[1]{\ensuremath{ \left \langle #1 \right \rangle }}

\newcommand{\p}{\mathbb{P}}
\newcommand{\q}{\mathbb{Q}}
\newcommand{\E}{\mathbb{E}}

\newcommand{\EQ}{\E^{\q}}
\newcommand{\R}{\mathbb{R}}

\newcommand{\bR}{\overline{\mathbb{R}}}

\newcommand{\lm}{\mathop{\mathrm{lim\,med}}}

\newcommand{\BIGOP}[1]{\mathop{\mathchoice%
    {\raise-0.22em\hbox{\huge $#1$}}%
    {\raise-0.05em\hbox{\Large $#1$}}{\hbox{\large $#1$}}{#1}}}

\newcommand{\BIGboxplus}{\mathop{\mathchoice%
    {\raise-0.35em\hbox{\huge $\boxplus$}}%
    {\raise-0.15em\hbox{\Large $\boxplus$}}{\hbox{\large $\boxplus$}}{\boxplus}}}

\begin{document}

\title{\vspace{-0em}
Robust expected utility maximization with medial limits
\date{November 2018}
\author{Daniel Bartl\thanks{
  Department of Mathematics and Statistics, University of Konstanz, 78464 Konstanz, Germany and
  Department of Mathematics, University of Vienna, 1090 Vienna, Austria.
  Financial support from the Austrian Science Fund (FWF) through grant Y00782 is gratefully acknowledged.
  }
  \and
  Patrick Cheridito\thanks{RiskLab, Department of Mathematics, ETH Zurich, 8092 Zurich, Switzerland.
   }
  \and
    Michael Kupper\thanks{Department of Mathematics and Statistics, University of Konstanz, 78464 Konstanz, Germany}}
}

\maketitle \vspace{-1.2em}

\begin{abstract}
In this paper we study a robust expected utility maximization problem with random endowment in 
discrete time. We give conditions under which an optimal strategy exists and 
derive a dual representation for the optimal utility. Our approach is based on a
general representation result for monotone convex 
functionals, a functional version of Choquet's capacitability theorem and medial limits. 
The novelty is that it works under nondominated model uncertainty without any assumptions 
of time-consistency. As applications, we discuss robust utility maximization problems with 
moment constraints, Wasserstein constraints and Wasserstein penalties.
\end{abstract}

\vspace{.9em}

{\small
}

\section{Introduction}
\label{sec:intro}

We consider a robust expected utility maximization problem of the form
\be \label{U1}
U(X) = \sup_{\vartheta\in\Theta} \inf_{\mathbb{P} \in \mathcal{P}} \crl{\mathbb{E}^\mathbb{P}
u \brak{X+ \sum_{t=1}^T \vartheta_t \Delta S_t} + \alpha(\mathbb{P})},
\ee
where $X$ is a random endowment, $S_0, S_1, \dots, S_T$ the price evolution of a tradable asset, 
$\Theta$ the set of possible trading strategies, $u$ a random utility function, 
${\cal P}$ a set of probability measures and $\alpha \colon {\cal P} \to [0,\infty)$ a penalty function.
In the special case $\alpha \equiv 0$, \eqref{U1} reduces to 
\be \label{U2}
U(X) = \sup_{\vartheta\in\Theta} \inf_{\mathbb{P} \in \mathcal{P}} \mathbb{E}^\mathbb{P}
u \brak{X+ \sum_{t=1}^T \vartheta_t \Delta S_t}.
\ee

A large strand of the literature on robust utility maximization assumes that 
the family ${\cal P}$ is dominated\footnote{i.e., all $\p \in {\cal P}$ are absolutely continuous with 
respect to a common probability measure $\mathbb{P}^*$}; see e.g. 
\cite{HS, quenez2004optimal, gundel2005utility, S05, S07, BMS, owari2011robust, backhoff2015robust}.
In this case, one can, as in the classical expected utility framework ${\cal P} = \crl{\p}$, apply Koml\'os' theorem to construct an 
optimal strategy from a sequence of approximately optimal strategies. The existence of optimal 
strategies can then be used to deduce a dual representation for $U$. Different discrete-time versions of 
problem \eqref{U2} under nondominated model uncertainty have been 
studied by \cite{N16,BC,NS,B}. They all make time-consistency assumptions\footnote{Problem 
\eqref{U2} is time-consistent if the set ${\cal P}$ is stable under concatenation of 
transition probabilities. Conditions for time-consistency of problems of the form \eqref{U1} are given in 
e.g. \cite{CDK, CK}.}, which allows them to tackle the problem step by step backwards in time using
dynamic programming arguments. In continuous time, nondominated problems of the form \eqref{U2} have been investigated by 
\cite{MPZ, NN} in the case, where ${\cal P}$ consists of a time-consistent family of martingale or L\'evy process laws. 

In this paper we study problem \eqref{U1} without domination or time-consistency assumptions.
As a consequence, we cannot apply Koml\'os' theorem or dynamic programming arguments. 
Instead, we use convex duality methods, a functional version of Choquet's capacitability theorem \cite{Ch}
and medial limits. For our purposes, a medial limit is a positive linear functional $\lm \colon l^{\infty} \to \mathbb{R}$ satisfying
$\liminf \le \lm \le \limsup$ with the following property: for any uniformly bounded sequence of universally 
measurable\footnote{Recall that the universal completion $\mathcal{F}^\ast$ of a 
$\sigma$-algebra ${\cal F}$ is defined as the intersection of $\sigma(\mathcal{F}\cup\mathcal{N}^\mathbb{P})$
over all probability measures $\p$ on $\mathcal{F}$, where $\mathcal{N}^\mathbb{P}$ denotes the collection of
$\mathbb{P}$-null sets. By saying that $X \colon E \to \mathbb{R}$ is universally measurable,
we mean that it is measurable with respect to the universal completion ${\cal F}^*$ of ${\cal F}$ 
and the Borel $\sigma$-algebra on $\mathbb{R}$, which is equivalent to saying that
$X$ is measurable with respect to ${\cal F}^*$ and the universal completion of the Borel 
$\sigma$-algebra on $\mathbb{R}$.} functions $X_n : E \to \mathbb{R}$ on a measurable space
$(E, {\cal F})$, $X = \lm X_n$ is universally measurable and 
$\mathbb{E}^{\mathbb{P}} X = \lm \mathbb{E}^{\mathbb{P}} X_n$ for every 
probability measure $\mathbb{P}$ on the universal completion of ${\cal F}$.
Mokobozki proved that medial limits exist under the usual axioms of ZFC 
together with the continuum hypothesis; see \cite{Meyer}. 
Later, Normann \cite{Normann} showed that it is enough to assume ZFC and Martin's axiom.
In \cite{N14} medial limits were used to establish the existence of optimal quasi-sure superhedging
strategies with respect to general sets of martingale measures.

We first derive a dual representation of $U(X)$ for lower semicontinuous random endowments $X$ only from
convexity and integrability assumptions. Then we show that a suitable no-arbitrage condition
and the existence of a medial limit imply that problem \eqref{U1} admits optimal strategies. From there
we can extend the dual representation of $U(X)$ from lower semicontinuous to measurable random endowments $X$.

As sample space we consider a non-empty subset $\Omega$ of $((0,\infty) \times \mathbb{R})^{T+1}$
endowed with the Euclidean metric and the corresponding Borel $\sigma$-algebra.
We suppose there is a money market account evolving according to $M_t(\omega) = \omega_{t,1}$ and a 
financial asset whose price in units of $M_t$ is given by $S_t(\omega) = \omega_{t,2}$. 
$X \colon \Omega \to \mathbb{R}$ is a Borel measurable mapping describing a random endowment in units of $M_T$.
As usual, $\Delta S_t$ denotes the increment $S_t - S_{t-1}$.
$\mathcal{P}$ is assumed to be a non-empty set of Borel probability measures on $\Omega$ and 
$\alpha \colon {\cal P} \to \mathbb{R}_+ := [0,\infty)$ a mapping with the property $\inf_{\p \in {\cal P}} \alpha(\p) = 0$.
Denote by $({\cal F}_t)_{t=0}^T$ the filtration generated by $(M_t,S_t)_{t=0}^T$.
The set $\Theta$ consists of all strategies $(\vartheta_t)_{t=1}^T$ such that for each $t$,
$\vartheta_t \colon \Omega \to \R$ is measurable with respect to the 
universal completion ${\cal F}^*_{t-1}$ of ${\cal F}_{t-1}$ and the Borel $\sigma$-algebra on $\mathbb{R}$. 
$u \colon\Omega\times\mathbb{R}\to\mathbb{R}$ is a random utility function, which we assume to
satisfy the following conditions:
\begin{itemize}
\item[(U1)] $u(\omega,x)$ is increasing\footnote{In the whole paper we understand the words
``increasing'' and ``decreasing'' in the weak sense. That is, $u$ satisfies $u(\omega,x) \ge u(\omega,y)$ 
for all $x \ge y$.} and concave in $x$
\item[(U2)] for every $n \in \mathbb{N}$, $u\colon\Omega\times [-n,\infty)\to\mathbb{R}$ is 
continuous and bounded
\item[(U3)] $\lim_{x \to - \infty} \sup_{\omega\in \Omega} u(\omega,x)/|x| = -\infty$.
\end{itemize}
Note that if $u$ does not depend on $\omega$, \eqref{U1} measures the utility of the
discounted terminal wealth $X + \sum_{t=1}^T \vartheta_t \Delta S_t$. On the other hand, if 
$u$ is of the form $u(\omega,x) = \tilde{u}(\omega^1_T x)$ for a function 
$\tilde{u} \colon \mathbb{R} \to \mathbb{R}$, then \eqref{U1} evaluates the undiscounted 
terminal wealth $M_T (X + \sum_{t=1}^T \vartheta_t \Delta S_t)$.

We suppose there exists a continuous function $Z : \Omega \to [1,\infty)$ such that 
$Z \ge 1\vee \sum_{t=0}^T |S_t|$ and all sublevel sets 
$\crl{\omega \in \Omega : Z(\omega) \le z}$, $z \in \mathbb{R}_+$, are compact. Let $B_Z$ be the space of all 
Borel measurable functions $X : \Omega \to\mathbb{R}$ such that $X/Z$ is bounded, $L_Z$ the 
set of all lower semicontinuous $X \in B_Z$ and $C_Z$ the space of all continuous $X \in B_Z$.
By ${\cal M}_Z$ we denote the set of all Borel probability measures $\p$ on $\Omega$ 
satisfying $\mathbb{E}^{\p} Z < \infty$. Then 
$\mathbb{E}^{\mathbb{P}} X$ is well-defined for all $\p \in\mathcal{M}_Z$ and $X\in B_Z$.

To derive dual representations for $U$, we need ${\cal P}$ and $\alpha$ to satisfy the following two conditions:

\begin{itemize}
\item[{\rm (A1)}] ${\cal P}$ is a convex subset of ${\cal M}_Z$ and $\alpha \colon {\cal P} \to \mathbb{R}_+$ a
convex mapping with $\sigma({\cal M}_Z,C_Z)$-closed sublevel sets 
${\cal P}_c := \crl{\p \in {\cal P} : \alpha(\p) \le c}$, $c \in \mathbb{R}_+$ 
\item[{\rm (A2)}] there exists an increasing function $\beta \colon [1, \infty) \to \mathbb{R}$ 
such that $\lim_{x \to \infty} \beta(x)/x = \infty$ and  
\[\inf_{\mathbb{P} \in \mathcal{P}} \crl{\mathbb{E}^\mathbb{P}u(-\beta(Z)) + \alpha(\p)} > - \infty.\]
\end{itemize}
By $v$ we denote the convex conjugate of $u$, given by 
\[
v(\omega,y):=\sup_{x\in\mathbb{R}} \crl{u(\omega,x)-xy}, \quad (\omega, y) \in \Omega \times \mathbb{R}_+.
\]
If $u$ satisfies (U2), $u(\omega,0)$ is bounded in $\omega$, and one has 
\[
v(\omega,y) = \sup_{x\in\mathbb{Q}} \crl{u(\omega,x)-xy} \ge u(\omega, 0).
\] 
In particular, $v$ is a Borel measurable function from $\Omega \times \mathbb{R}_+$ to  $(-\infty,\infty]$ 
that is bounded from below. So for $q \in \mathbb{R}_+$ and a Borel probability measure $\q$ on $\Omega$, 
one can define
$$
D^{\alpha}_v(q\q) := \inf_{\p \in {\cal P}} \crl{D_v(q\q \,\| \,\p) + \alpha(\mathbb{P})},
$$
where $D_v(q\q \,\| \, \p)$ is the $v$-divergence between $q\q$ and $\p$, given by 
$$
D_v(q\q \,\| \, \p) := \left\{\begin{array}{ll}
\mathbb{E}^{\mathbb{P}}v\big(qd\mathbb{Q}/d\mathbb{P}\big) & \quad \mbox{if } q\q \ll \p \\
\infty & \quad \mbox{otherwise.}
\end{array} \right.
$$
Let ${\cal Q}_Z$ be the set of all probability measures $\p \in {\cal M}_Z$ under which $(S_t)_{t=0}^T$
is a martingale and $\hat{\cal Q}_Z$ the set of all pairs $(q,\q) \in \mathbb{R}_+ \times {\cal M}_Z$
such that $q = 0$ or $\q \in {\cal Q}_Z$. Our first duality result is as follows:

\begin{theorem} \label{thm1}
Assume {\rm (U1)--(U3)} and {\rm (A1)--(A2)}. Then
\be \label{dualL}
U(X) = \min_{(q, \mathbb{Q}) \in \hat{\cal Q}_Z} 
\crl{q \mathbb{E}^{\mathbb{Q}}X + D^{\alpha}_v(q \q)} \in \mathbb{R}
\quad \mbox{for all } X \in L_Z.
\ee
\end{theorem}

To be able to derive the existence of optimal strategies and extend the duality \eqref{dualL} 
to Borel measurable random endowments $X$, we need the following no-arbitrage condition\footnote{Obviously,
(NA) is weaker than the assumption that no $\p \in {\cal P}$ admits arbitrage. On the other hand, 
it implies e.g. the robust no-arbitrage condition NA$({\cal P})$ of \cite{BN}, which has been 
used in \cite{N16,BC,NS,B} to derive the existence of optimal strategies. Indeed,
assume (NA) holds and there exists a strategy such that $\p[\sum_{t=1}^T \vartheta_t \Delta S_t \ge 0] = 1$
for all $\p \in {\cal P}$. Then each $\p \in {\cal P}$ is dominated by a $\p' \in {\cal P}$ that does not 
admit arbitrage. Hence, $\p[\sum_{t=1}^T \vartheta_t \Delta S_t > 0] = \p'[\sum_{t=1}^T \vartheta_t \Delta S_t > 0] =0$, 
showing that NA$({\cal P})$ holds.}:

\begin{itemize}
\item[{\rm (NA)}] every $\mathbb{P} \in \mathcal{P}$ is dominated by a $\mathbb{P}' \in\mathcal{P}$ that 
does not admit arbitrage,
\end{itemize}
where a Borel probability measure $\p$ on $\Omega$ is said to admit arbitrage if there exists a strategy 
$\vartheta \in \Theta$ such that $\p[\sum_{t=1}^T \vartheta_t \Delta S_t > 0] > 0$ and 
$\p[\sum_{t=1}^T \vartheta_t \Delta S_t \ge 0] = 1$. 

\begin{theorem} \label{thm2}
Assume a medial limit exists, $u$ fulfills {\rm (U1)}--{\rm (U3)} and {\rm (NA)} holds.
Then the supremum in \eqref{U1} is attained for every Borel measurable function $X \colon \Omega \to \mathbb{R}$ 
such that $U(X) \in \mathbb{R}$. If, in addition, {\rm (A1)--(A2)} are satisfied, then 
\be \label{dualB}
U(X) = \inf_{(q, \mathbb{Q}) \in \hat{\cal Q}_Z} 
\crl{q \mathbb{E}^{\mathbb{Q}}X + D^{\alpha}_v(q \q)} \in \mathbb{R} \quad \mbox{for all } X \in B_Z.
\ee
\end{theorem}

In the special case, where $\alpha \equiv 0$ and $u$ is of the form  $u(x) = - \exp(-\lambda x)$ for a risk-aversion 
parameter $\lambda > 0$, the dual expression \eqref{dualB} simplifies if instead of \eqref{U2}, 
one considers the equivalent problem
$$
W(X) = \sup_{\vartheta \in \Theta} \inf_{\mathbb{P} \in \mathcal{P}} - \frac{1}{\lambda} \log \mathbb{E}^\mathbb{P}
\exp \brak{- \lambda X - \lambda \sum_{t=1}^T \vartheta_t \Delta S_t}.
$$

\begin{corollary} \label{cor}
Assume a medial limit exists and
${\cal P}$ is a non-empty $\sigma({\cal M}_Z, C_Z)$-closed convex subset of ${\cal M}_Z$ 
satisfying {\rm (NA)}. If there exists an increasing function $\beta \colon [1,\infty) \to \mathbb{R}$ 
such that $\lim_{x \to \infty} \beta(x)/x = \infty$ and
\[
\sup_{\p \in {\cal P}} \mathbb{E}^{\p} \exp(\beta(Z)) < \infty,
\] 
then
\[
W(X) = \inf_{\mathbb{Q} \in \mathcal{Q}_Z} 
\crl{\mathbb{E}^{\mathbb{Q}}X + \frac{1}{\lambda} H(\q \,\| \, {\cal P})} \in \mathbb{R} \quad \mbox{for all } X \in B_Z,
\]
where $H(\q \,\| \, {\cal P}) := \inf_{\p \in {\cal P}} H(\q \,\| \, \p)$ is the robust version of the relative entropy 
\[
H(\q \,\| \, \p) := \begin{cases}
\EQ \log(d\q / d\p) & \quad \mbox{if } \q \ll \p \\
 \infty & \quad \mbox{otherwise.}
\end{cases}
\]
\end{corollary}

In the following, we discuss three examples of robust utility maximization problems 
that are neither dominated nor time-consistent but still fit in our framework.

\begin{Example} \label{ex:mo}
Our first example is of the form \eqref{U2} for a set of probability measures ${\cal P}$ 
given by moment constraints. Consider a sample space of the form 
$\Omega = \Omega_0 \times \dots \times \Omega_T$, where $\Omega_0 = \crl{(a_0,s_0)}$ for 
fixed initial values $a_0,s_0 > 0$ and $\Omega_t = [a_t,b_t] \times (0,\infty)$ for constants $0 < a_t \leq b_t$,
$t = 1, \dots, T$. Note that $Z(\omega) = \sum_{t=0}^T \omega_{t,2} \vee (\omega_{t,2})^{-1}$ defines a continuous 
function $Z \colon \Omega\to [1, \infty)$ with compact sublevel sets $\crl{Z \le z}$, $z \in \mathbb{R}_+$, 
such that $Z \ge 1 \vee \sum_{t=0}^T |S_t|$. For all $t = 1, \dots, T$ and $i = 1, \dots, I$, 
let $c^i < 0$ and $d^i, C^i_t, D^i_t > 0$ be constants such that $\min_i c^i < -1$ and $\max_i d^i > 1$. 
Assume that the set ${\cal P}$ of all Borel probability measures on $\Omega$ satisfying the 
moment constraints\footnote{Alternatively, one can consider a set ${\cal P}$ of 
Borel probability measures satisfying moment conditions of the form
$\mathbb{E}^\mathbb{P}[(M_t S_t)^{c^i}] \le C^i_t$
and $\mathbb{E}^\mathbb{P}[(M_tS_t)^{d^i}] \le D^i_t$ for $t = 1, \dots, T$ and
$i = 1, \dots, I$, where $M_t$ describes the evolution of the money market account. Then, 
provided that ${\cal P}$ is non-empty, (A1) is still satisfied, and (A2) holds under the same conditions on $u$.
A sufficient condition for (NA) is that there exist constants $e_t \in [a_t,b_t]$ such that $(e_ts_0)^{c^i} < C^i_t$ and 
$(e_ts_0)^{d^i} < D^i_t$ for all $t = 1, \dots, T$ and $i = 1, \dots, I$.
}
\[
\mathbb{E}^\mathbb{P}[S_t^{c^i}] \le C^i_t \quad \mbox{and} \quad 
\mathbb{E}^\mathbb{P}[S_t^{d^i}] \le D^i_t \quad \mbox{for all } t=1, \dots, T \mbox{ and } i = 1, \dots, I,
\]
is non-empty. Then ${\cal P}$ fulfills (A1) for $\alpha \equiv 0$. 
Moreover, if $u \colon \Omega\times \mathbb{R} \to \mathbb{R}$ is a random 
utility function satisfying (U1)--(U3) and there exists a constant 
\[
q < \max_{1 \le i \le I} |c^i| \wedge \max_{1 \le i \le I} |d^i|
\] 
such that $u(\omega, x)/(1+ |x|^q)$ is bounded, then (A2) holds for $\alpha \equiv 0$. 
Finally, if $s_0^{c^i} < C^i_t$ and $s_0^{d^i} < D^i_t$ for all $t = 1, \dots, T$ and $i=1, \dots,I$,
then ${\cal P}$ also satisfies (NA). Proofs are given in Appendix \ref{ap:mo}.

\end{Example}

\begin{Example} \label{ex:wball}
As a second example, we consider a problem of the form \eqref{U2} with a set ${\cal P}$ of probability measures 
that are within a given Wasserstein distance of a reference measure. Let the sample space 
$\Omega$ be of the same form as in Example \ref{ex:mo}, and consider the metric
\[
d(\omega, \omega') := \brak{\sum_{t =1}^T e^{-\rho \kappa t} (|\omega_{t,1} - \omega'_{t,1}|^{\kappa} 
+ |\varphi(\omega_{t,2}) - \varphi(\omega'_{t,2})|^{\kappa})}^{1/\kappa}, \quad 
\omega, \omega' \in \Omega,
\]
where $\rho \ge 0$ and $\kappa \ge 1$ are constants and the function $\varphi \colon (0, \infty) \to \mathbb{R}$
is given by 
\[
\varphi(x) := \begin{cases}
x-1 & \mbox{ if } x > 1\\
\log(x) & \mbox{ if } x \le 1.
\end{cases}
\]
Denote $\omega^* = ((a_0,s_0), (a_1,1), \dots, (a_T,1)) \in \Omega$. Then,
$Z(\omega) = s_0 + T + e^{\rho T} T^{1-1/\kappa} d(\omega, \omega^*)$ is a continuous function 
$Z \colon \Omega\to [1, \infty)$ with compact sublevel sets $\crl{Z \le z}$, $z \in \mathbb{R}_+$, such 
that $Z \ge 1 \vee \sum_{t=0}^T |S_t|$. Choose a reference measure $\p^* \in {\cal M}_Z$ satisfying 
$\mathbb{E}^{\p^*} Z^p < \infty$ for a given exponent $p > 1$.
Fix a constant $\eta > 0$, and consider the ball 
\[ 
{\cal P}:= \crl{\p \in\mathcal{M}_Z :  W_p(\p,\p^*) \leq \eta}
\]
around $\p^*$ with respect to the $p$-Wasserstein distance $W_p$, given by 
\[
W_p(\p, \p^*) := \inf_{\pi} \brak{\int_{\Omega \times \Omega} d(\omega, \omega')^p d \pi(\omega, \omega')}^{1/p},
\] 
where the infimum is taken over all Borel probability measures $\pi$ on $\Omega \times \Omega$
with marginals $\p$ and $\p^*$. Then ${\cal P}$ satisfies (A1) for $\alpha \equiv 0$ as well as (NA). Moreover, if 
$u \colon \Omega\times \mathbb{R} \to \mathbb{R}$ is a random 
utility function satisfying (U1)--(U3) and there exists a constant $q < p$ such that 
$u(\omega, x)/(1+|x|^q)$ is bounded, then also (A2) holds for $\alpha \equiv 0$. 
This is proved in Appendix \ref{ap:wball}.
\end{Example}

\begin{Example}
\label{ex:wpen}
As our last example, we consider a problem of the form \eqref{U1} with a Wasserstein penalty.
Let the sample space $\Omega$ be of the same form as in Examples \ref{ex:mo} and \ref{ex:wball}. 
Fix an exponent $p > 1$, and let $Z$, $d$, $W_p$ be as in Example \ref{ex:wball}. For a given constant $\eta > 0$
and a reference measure $\p^* \in {\cal M}_Z$ satisfying $\mathbb{E}^{\p^*} Z^p < \infty$, define
$\alpha(\mathbb{P}) := \eta W_p(\p,\p^*)^p$ and ${\cal P} := \crl{\p \in {\cal M}_Z : \alpha(\p) < \infty}$.
Then (A1) and (NA) hold. Moreover, if $u \colon \Omega\times \mathbb{R} \to \mathbb{R}$ is a random 
utility function satisfying (U1)--(U3) and there exists a constant $q < p$ such that 
$u(x)/(1+|x|^q)$ is bounded, then (A2) is fulfilled as well. Proofs are provided in Appendix \ref{ap:wpen}.
\end{Example}

The rest of the paper is organized as follows. In Section \ref{sec:rep} we first establish a functional 
version of Choquet's capacitability theorem. Then we derive dual representation results for 
increasing convex functionals on different sets of real-valued functions. 
These results hold for general sample spaces endowed with a 
perfectly normal topology\footnote{in particular, for metrizable sample spaces} 
and do not require the existence of a medial limit. In Section \ref{sec:proofs}, we first prove 
Theorem \ref{thm1}. Then we derive some elementary properties of medial limits,
before we give proofs of Theorem \ref{thm2} and Corollary \ref{cor}. In the appendix we show 
that conditions (A1), (A2) and (NA) hold in the three Examples \ref{ex:mo}, \ref{ex:wball} and \ref{ex:wpen}.

\setcounter{equation}{0}
\section{Functional version of Choquet's capacitability theorem and dual representation of increasing convex functionals}
\label{sec:rep}

In this section, we first derive a functional version of Choquet's capacitability theorem by working 
out a remark at the end of his paper \cite{Ch}. Then we establish a dual representation 
result for increasing convex functionals defined on spaces of measurable functions.

Denote by $\bR$ the extended real line $[-\infty, \infty]$. 
For a given non-empty set $E$, consider two nested subsets $H \subseteq G \subseteq \bR^{E}$ 
such that $H$ is a non-empty lattice and $G$ contains all suprema of increasing\footnote{We 
call a sequence $(X_n)$ in $G$ increasing if $X_{n+1} \ge X_n$ for all $n$ and 
decreasing if $X_{n+1} \le X_n$ for all $n$.} sequences in $G$ as well as
all infima of arbitrary sequences in $G$. An $H$-Suslin scheme is a mapping 
$\sigma : \bigcup_{n\in\mathbb{N}}\mathbb{N}^n\to H$ and an 
$H$-Suslin function an element $X \in \bR^{E}$ of the form 
$$
X = \sup_{\gamma\in\mathbb{N}^{\mathbb{N}}}\inf_{n\in\mathbb{N}} \sigma(\gamma_1,\dots,\gamma_n),
$$ 
where $\sigma$ is an $H$-Suslin scheme. We denote the set of all $H$-Suslin functions by $S(H)$ 
and all infima of sequences in $H$ by $H_{\delta}$. If $\phi : G \to \bR$ is an increasing\footnote{that is,
$\phi(X) \ge \phi(Y)$ for all $X,Y \in G$ such that $X \ge Y$} mapping, 
we extend it to $\bR^{E}$ by setting
$$
\hat{\phi}(X) := \inf \crl{\phi(Y) : X \le Y, \, Y \in G}, \; X \in \bR^{E} 
\quad \mbox{with the convention } \inf \emptyset := + \infty.$$

The following is a functional version of Theorem 1 in \cite{Ch}:

\begin{proposition} \label{prop:capac}
Let $\phi : G \to \bR$ be an increasing mapping with the following two properties:
  
\begin{itemize}
\item[{\rm (C1)}] $\lim_n \phi(X_n)= \phi(\lim_n X_n)$ for every decreasing sequence $(X_n)$ in $H$
\item[{\rm (C2)}] $\lim_n \phi(X_n) = \phi(\lim_n X_n)$ for every increasing sequence $(X_n)$ in  $G$.
\end{itemize}
Then, $\hat{\phi}(X)=\sup\{\phi(Y): Y\leq X, \, Y\in H_\delta\}$ for all $X\in S(H)$.
\end{proposition}

\begin{proof}
Denote $F = E \times \bR$, and let ${\cal A}$ be the collection of subsets of $F$
of the form $\bigcup_{x \in E} \crl{x} \times A_x$, where for each $x$, $A_x =[-\infty,a_x)$ or 
$A_x =[-\infty,a_x]$ for some $a_x \in \bR$. 
Then ${\cal A}$ is stable under intersections and unions. For $A \in {\cal A}$, define 
$X_A : E \to \bR$ by $X_A(x) := a_x$. Then for any 
family of subsets $(A_{\alpha}) \subseteq {\cal A}$, one has $X_{\bigcap_{\alpha} A_{\alpha}} = \inf_{\alpha} X_{A_{\alpha}}$
and $X_{\bigcup_{\alpha} A_{\alpha}} = \sup_{\alpha} X_{A_{\alpha}}$. In particular,
$\mathcal{H}_{\delta} := \{A \in \mathcal{A}: X_A \in H_{\delta}\}$ is stable under finite unions and countable intersections.
It is clear that the set function $\tilde{\phi} : 2^F \to \bR$, given by 
$$\tilde{\phi}(B) := \inf\{\hat{\phi}(X_A) : B \subseteq A, \, A \in \mathcal{A} \},$$ is increasing\footnote{that is,
$\tilde{\phi}(B) \ge \tilde{\phi}(C)$ for all $B,C \in 2^F$ such that $B \supseteq C$} and satisfies 
$\lim_n \tilde{\phi}(B_n) = \tilde{\phi}(\bigcap_n B_n)$ for decreasing\footnote{that is, $B_{n+1} \subseteq B_n$
for all $n$} sequences $(B_n)$ in 
${\cal H}_{\delta} $. Moreover, if $(B_n)$ is an increasing\footnote{that is, $B_{n+1} \supseteq B_n$ for all $n$} 
sequence of subsets of $F$ such that $\lim_n \tilde{\phi}(B_n) < + \infty$, there exist $A_n \in {\cal A}$ and $Y_n \in G$ such that
$B_n \subseteq A_n$, $X_{A_n} \le Y_n$ and $\phi(Y_n) \le \tilde{\phi}(B_n) + 1/n$ (or $\phi(Y_n) \le -n$ in case
$\tilde{\phi}(B_n) = - \infty$). The sequences $\tilde{A}_n = \bigcap_{m \ge n} A_m$ and $\tilde{Y}_n = \inf_{m \ge n} Y_m$
are increasing, and one has $\bigcup_n B_n \subseteq A := \bigcup_n \tilde{A}_n \in {\cal A}$ 
as well as $X_A \le Y := \sup_n \tilde{Y}_n \in G$. So 
$$
\mbox{$\tilde{\phi}\brak{\bigcup_n B_n}$} \le \hat{\phi}(X_A) \le \phi(Y) = \lim_n \phi(\tilde{Y}_n) \le \lim_n \tilde{\phi}(B_n).
$$
This shows that $\tilde{\phi}$ is an abstract capacity on $(F, {\cal H}_{\delta})$
according to \cite{Ch}. For an $H$-Suslin function of the form
$X = \sup_{\gamma \in \mathbb{N}^{\mathbb{N}}}\inf_{n \in \mathbb{N}} \sigma(\gamma_1,\dots,\gamma_n)$,
define $\tilde{\sigma} : \bigcup_{n \in \mathbb{N}}\mathbb{N}^n\to\mathcal{H}_{\delta}$ by
$\tilde{\sigma}(\cdot) := \bigcup_{x \in E} \crl{x} \times [-\infty, \sigma(\cdot)(x)]$. Then
\[A = \bigcup_{\gamma \in \mathbb{N}^{\mathbb{N}}} \bigcap_{n \in\mathbb{N}}
\tilde{\sigma}(\gamma_1,\dots,\gamma_n)\] is a Suslin set generated by ${\cal H}_{\delta}$ satisfying $X_A =X$.
So one obtains from Theorem 1 of \cite{Ch} that
$$
\hat{\phi}(X) = \tilde{\phi}(A) = \sup\{\tilde{\phi}(B): B \subseteq A,  \, B \in {\cal H}_\delta\} 
= \sup\{\phi(Y): Y \leq X, \, Y \in H_\delta\}.
$$
\end{proof}

In the following, let $E$ be a perfectly normal topological space\footnote{In particular, this covers all metric spaces.} 
and $V \colon E \to \mathbb{R}_+ \setminus \crl{0}$ 
a continuous function. Denote by $B_V$ the set of all Borel measurable functions 
$X\colon E\to\mathbb{R}$ such that $X/V$ is bounded and
by $C_V$ and $U_V$ the subsets consisting of all continuous and upper semicontinuous functions in $B_V$, respectively.
If $(X_n)$ is an increasing (decreasing) sequence of real-valued functions on $E$ that converges pointwise to a real-valued 
function $X$ on $E$, we write $X_n \uparrow X$ ($X_n \downarrow X$). Let $ca^+_V$ be the set of all Borel 
measures $\mu$ on $E$ satisfying $\ang{V, \mu} < + \infty$. For a real-valued mapping $\phi$ 
defined on a subset of $B_V$ containing $C_V$, we define
\be \label{conj}
\phi^*_{C_V}(\mu):= \sup_{X \in C_V} \crl{\ang{X,\mu}-\phi(X)}, \quad \mu \in ca^+_V.
\ee
Then the following holds:

\begin{theorem} \label{thm:dual}
If $\phi : C_V \to \mathbb{R}$ is an increasing convex functional satisfying 

\begin{itemize}
\item[{\rm (R1)}]
$\phi(X_n)\downarrow \phi(0)$ for every sequence $(X_n)$ in $C_V$ such that $X_n \downarrow 0$,
\end{itemize}
then \be \label{repC} 
\phi(X) = \max_{\mu \in ca_V^+} \crl{\ang{X,\mu} - \phi_{C_V}^\ast(\mu)} \quad \text{for each } X \in C_V,
\ee
and all sublevel sets $\{\mu \in ca^+_V : \phi^*_{C_V}(\mu)\leq c \},$ $c \in \mathbb{R}$, are 
$\sigma(ca^+_V, C_V)$-compact.

Moreover, every increasing convex functional $\phi : U_V \to \mathbb{R}$ with the property 
\begin{itemize}
\item[{\rm (R2)}] 
$\phi(X_n) \downarrow \phi(X)$ for each sequence $(X_n)$ in $C_V$ such that $X_n\downarrow X$
for some $X \in U_V$,
\end{itemize}
has a representation of the form 
\be \label{repU} \phi(X)=\max_{\mu \in ca_V^+} \crl{\ang{X,\mu} - \phi_{C_V}^\ast(\mu)}, \quad
X\in U_V,\ee
and every increasing convex functional $\phi : B_V \to \mathbb{R}$ satisfying {\rm (R2)} together with
\begin{itemize}
\item[{\rm (R3)}] 
$\phi(X_n)\uparrow\phi(X)$ for each sequence $(X_n)$ in $B_V$ such that $X_n \uparrow X$ for some $X \in B_V$,
\end{itemize}
can be written as
\be \label{repB} \phi(X) = \sup_{\mu \in ca_V^+} \crl{\ang{X,\mu} - \phi_{C_V}^\ast(\mu)}, \quad
X\in B_V.\ee
\end{theorem}

\begin{proof}
First, let $\phi \colon C_V \to \mathbb{R}$ be an increasing convex functional satisfying (R1).
It is clear from the definition of $\phi^*_{C_V}$ that for fixed $X\in C_V$,
\be \label{ineq}
\phi(X) \ge \ang{X,\mu} - \phi^*_{C_V}(\mu) \quad \mbox{for all } \mu \in ca^+_V.
\ee
Moreover, it follows from the Hahn--Banach extension theorem that there exists a positive linear functional 
$\psi \colon C_V \to \mathbb{R}$ such that 
\[\psi(Y) \le \phi(X+Y) - \phi(X)\quad\mbox{for all }Y \in C_V.\]
Now, consider a sequence $(X_n)$ of functions in $C_V$ such that $X_n \downarrow 0$. Then,
one has for all $\lambda \in (0,1)$, 
\be \label{conv}
\phi(X+X_n) \le \lambda \phi\Big(\frac{X}{\lambda} \Big)+(1-\lambda)\phi \Big(\frac{X_n}{1-\lambda}\Big).
\ee
Since $y \mapsto \phi(y X)$ is a convex function from $\mathbb{R}$ to $\mathbb{R}$, it is continuous.
Therefore, for $\lambda$ close to $1$, $\lambda \phi(X/\lambda)$ is close to $\phi(X)$. By (R1), one has
$(1-\lambda) \phi(X_n/(1-\lambda)) \downarrow (1-\lambda) \phi(0)$. It follows that 
$\phi(X + X_n) \downarrow \phi(X)$, and consequently, $\psi(X_n) \downarrow 0$ for $n \to + \infty$. 
Since on a perfectly normal space, the Borel $\sigma$-algebra coincides with the $\sigma$-algebra 
generated by all continuous real-valued functions (see \cite{HT}), one obtains from 
the Daniell--Stone theorem that there exists a $\mu \in ca^+_V$ such that 
$\psi(Y) = \ang{Y,\mu}$ for all $Y \in C_V$. Hence,
\[
\ang{X+Y,\mu} -\phi(X+Y) \le \ang{X,\mu} - \phi(X) \quad\mbox{for all }Y\in C_V.\]
In particular, $\phi^*_{C_V}(\mu) = \ang{X,\mu} - \phi(X)$, which together with \eqref{ineq}, proves \eqref{repC}.

Next, we show that the sublevel sets 
$$
\Lambda_c :=\{\mu \in ca^+_V : \phi^*_{C_V}(\mu)\leq c \}, \quad c \in \mathbb{R},
$$
are $\sigma(ca^+_V,C_V)$-compact. Note that $C_V$ equipped with the norm $\|X\|_V := \sup_x |X(x)/V(x)|$
is a Banach space. We extend $\phi^*_{C_V}$ to the positive cone $C_V^{*,+}$ in the topological dual 
$C^*_V$ of $C_V$ using definition \eqref{conj}. Then the set $\tilde{\Lambda}_c 
:= \{\mu \in C_V^{*,+} : \phi^*_{C_V}(\mu) \leq c\}$
is $\sigma(C_V^*,C_V)$-closed. Moreover, since $\phi$ is real-valued, the 
increasing convex function $\varphi : \mathbb{R}_+ \to (-\infty, \infty]$, given by 
$\varphi(y) := \sup_{x \in \mathbb{R}_+} \crl{xy - \phi(xV)}$, satisfies $\lim_{y \to + \infty} \varphi(y)/y = \infty$.
As a consequence, the right-continuous inverse $\varphi^{-1} : \mathbb{R} \to \mathbb{R}_+$ has the property
$\lim_{x \to + \infty} \varphi^{-1}(x)/x = 0$. Since 
\[\phi^*_{C_V}(\mu) \ge \sup_{x \in \mathbb{R}_+} \crl{\ang{xV,\mu}- \phi(xV)}
= \varphi(\ang{V,\mu}),\]
one obtains for $\mu \in \tilde{\Lambda}_c$, 
\[\|\mu\|_{C_V^*} =\ang{V,\mu}
\leq \varphi^{-1}(\phi^*_{C_V}(\mu)) \leq \varphi^{-1}(c) < \infty.\]
So it follows from the Banach--Alaoglu theorem that $\tilde{\Lambda}_a$ is $\sigma(C_V^\ast,C_V)$-compact. 
Now, choose a $\mu \in C_V^{*,+}$ with $\phi^*_{C_V}(\mu)< \infty$ and let $(X_n)$ be a sequence 
in $C_V$ such that $X_n \downarrow 0$. Then, for every constant $y > 0$,
$\phi^*_{C_V}(\mu) \geq \ang{yX_n, \mu} -\phi(yX_n)$, and therefore,
\[\ang{X_n,\mu}  \leq \frac{\phi(yX_n)}{y} + \frac{\phi^*_{C_V}(\mu)}{y}.\]
By (R1), one obtains $\ang{X_n,\mu} \downarrow 0$, and it follows from the Daniell--Stone theorem that $\mu$ is in 
$ca^+_V$. This shows that $\phi^*_{C_V}(\mu)= \infty$ for all $\mu \in C_V^{*,+} \setminus ca^+_V$. 
In particular, $\Lambda_c$ is equal to $\tilde{\Lambda}_c$ and therefore, $\sigma(ca^+_V,C_V)$-compact. 

Now, assume $\phi \colon U_V \to \mathbb{R}$ is an increasing convex functional with the property (R2).
To show that the dual representation \eqref{repC} extends from $C_V$ to $U_V$, 
we use that on a perfectly normal space, every upper semicontinuous function is the pointwise limit 
of a decreasing sequence of continuous functions (see \cite{HT}). As an easy consequence, every $X \in U_V$
can be written as the pointwise limit of a decreasing sequence $(X_n)$ in $C_V$. It follows from (R2) and the definition of 
$\phi^*_{C_V}$ that 
\be \label{ge}
\phi(X) = \lim_n \phi(X_n) \ge \lim_n \ang{X_n, \mu} - \phi^*_{C_V}(\mu) 
\ge \ang{X,\mu} - \phi^*_{C_V}(\mu) \quad \mbox{for all } \mu \in ca^+_V.
\ee
On the other hand, one obtains from \eqref{repC} that
$$
\phi(X) \le \phi(X_n) = \max_{\mu \in ca_V^+} \crl{\ang{X_n,\mu} - \phi^*_{C_V}(\mu)} \quad \mbox{for every } n.
$$
Since
\beas
\ang{X_n,\mu} - \phi^*_{C_V}(\mu) &\le& \ang{X_1,\mu} - \phi^*_{C_V}(\mu) \le
\|X_1\|_V \|\mu\|_{C^\ast_V} - \phi^*_{C_V}(\mu)\\ &\le&
 \|X_1\|_V \varphi^{-1}(\phi^*_{C_V}(\mu)) - \phi^*_{C_V}(\mu),
\eeas
this implies that there exists a level $c \in \mathbb{R}$ such that 
$$ \phi(X_n)= \max_{\mu\in \Lambda_c} \crl{\langle X_n,\mu\rangle -\phi^\ast_{C_V}(\mu)}
\quad\text{for all $n$.}
$$
Note that $\ang{ X_n,\mu} -\phi^\ast_{C_V}(\mu)$ is decreasing in $n$ as well as $\sigma(ca^+_V,C_V)$-upper
semicontinuous and concave in $\mu$. So it follows from the minimax result, Theorem 2 of \cite{KF}, and the
monotone convergence theorem that 
\beas && \phi(X) = \inf_n\phi(X_n) = \inf_n \max_{\mu\in \Lambda_a} \crl{\langle X_n,\mu\rangle -\phi^\ast_{C_V}(\mu)}\\
&=& \max_{\mu\in \Lambda_a}\inf_n \crl{\langle X_n,\mu\rangle -\phi^\ast_{C_V}(\mu)}
= \max_{\mu\in \Lambda_a} \crl{\langle X,\mu\rangle -\phi^\ast_{C_V}(\mu)},
\eeas
which together with \eqref{ge}, proves \eqref{repU}.
  
The last part of Theorem \ref{thm:dual} follows from Proposition \ref{prop:capac}. Indeed, if 
$\phi \colon B_V \to \mathbb{R}$ is an increasing convex functional satisfying (R2)--(R3), we
fix a constant $r > 0$ and let $G$ be the set of $X \in B_V$ satisfying $|X| \le r|V|$. Then, 
$\phi$, $G$ and $H = C_V \cap G$ satisfy the assumptions of Proposition \ref{prop:capac}. 
Moreover, $H_{\delta} = U_V \cap G$. So it follows from Proposition \ref{prop:capac} and \eqref{repU} that
$$
\phi(X) 
= \sup_{Y \le X, \, Y \in U_V \cap G} \phi(Y) 
=  \sup_{Y \le X, \, Y \in U_V \cap G} \max_{\mu \in ca^+_V} \crl{\ang{Y,\mu} - \phi^*_{C_V}(\mu)}
\quad \mbox{for all } X \in G \cap S(H).
$$
Since for fixed $\mu \in ca^+_V$, the mapping $X \mapsto \ang{X, \mu}$ together with $G$ and $H$
also satisfies the assumptions of Proposition \ref{prop:capac}, one has
$$ \phi(X) 
= \sup_{\mu \in ca^+_V} \sup_{Y \le X, \, Y \in U_V \cap G} \crl{\ang{Y,\mu} - \phi^*_{C_V}(\mu)}
= \sup_{\mu \in ca^+_V} \crl{\ang{X,\mu} - \phi^*_{C_V}(\mu)} \quad \mbox{for all } X \in G \cap S(H).
$$
So, if we can show that $G \subseteq S(H)$, the representation \eqref{repB} holds for all $X \in B_V$ since 
$r$ was arbitrary. To prove $G \subseteq S(H)$, we note that a function $X \in G$ can be written as 
$$X = \sup_q \crl{q V1_{\crl{X \ge qV}} - rV 1_{\crl{X < qV}}},
$$
where the supremum is taken over all rational numbers $q$ in $[-r,r]$. Since in a perfectly normal
space, open sets can be represented as countable unions of closed ones (see \cite{HT}), one obtains from 
Proposition 7.35 and Corollary 7.35.1 in \cite{BS} that the Suslin sets generated by the closed sets contain 
the Borel $\sigma$-algebra. Therefore, $\{X\geq qV\}$ is of the form $\bigcup_{\gamma \in \mathbb{N}^{\mathbb{N}}}
\bigcap_{n \in \mathbb{N}} \tilde{\sigma}(\gamma_1,\dots,\gamma_n)$
for a Suslin scheme $\tilde{\sigma}$ with values in the closed subsets of $E$. The mapping
$\sigma:= q V 1_{\tilde{\sigma}}- r V 1_{\tilde{\sigma}^c}$ takes values in $H_{\delta}$, and so,
$$
q V 1_{\crl{X \geq qV}}- r V 1_{\crl{X < qV}}
 =\sup_{\gamma \in \mathbb{N}^{\mathbb{N}}} \inf_{n \in \mathbb{N}}
 \sigma(\gamma_1,\dots,\gamma_n)
 $$
belongs to $S(H_{\delta}) = S(H)$. Moreover, $S(H)$ is stable under taking countable suprema. 
Therefore, $X\in S(H)$, and the proof is complete.
\end{proof}

The following result gives a dual condition for (R2) which will be useful in the proof of Theorem
\ref{thm1} below.

\begin{proposition} \label{prop:dual}
An increasing convex functional $\phi : U_V \to \mathbb{R}$ with the property {\rm (R1)} satisfies 
{\rm (R2)} if and only if 
\be \label{dualcond}
\phi^*_{C_V}(\mu) = \phi^*_{U_V}(\mu) := \sup_{X \in U_V} \crl{\ang{X,\mu} - \phi(X)} \quad \mbox{for all } \mu \in ca^+_V.
\ee
\end{proposition}

\begin{proof}
First, let us assume $\phi$ satisfies (R2). For a given $X \in U_Z$, there exists a sequence $(X_n)$
in $C_V$ such that $X_n \downarrow X$ (see \cite{HT}). By the monotone convergence theorem and (R2), one has 
$$
\ang{X_n,\mu} - \phi(X_n) \to \ang{X,\mu} - \phi(X).
$$
This shows that $\phi^*_{C_V}(\mu) = \phi^*_{U_V}(\mu)$ for all $\mu \in ca^+_V$.

Now, assume $\phi$ satisfies (R1) together with \eqref{dualcond} and let $(X_n)$ be a sequence in 
$C_V$ such that $X_n \downarrow X \in U_V$. It is immediate from the definition
of $\phi^*_{U_V}$ and \eqref{dualcond} that 
$$
\phi(X) \ge \sup_{\mu \in ca^+_V} \crl{\ang{X,\mu} - \phi^*_{U_V}(\mu)} = 
\sup_{\mu \in ca^+_V} \crl{\ang{X,\mu} - \phi^*_{C_V}(\mu)}.
$$
On the other hand, it follows from the arguments in the proof of Theorem \ref{thm:dual} that 
there exists a $\sigma(ca^+_V, C_V)$-compact convex subset $\Lambda$ of $ca^+_V$ such that 
$\phi(X_n) = \max_{\mu \in \Lambda} (\ang{X_n,\mu} - \phi^*_{C_V}(\mu))$ for all $n$.
An application of the minimax result, Theorem 2 of \cite{KF}, and the monotone convergence theorem gives
\beas
&& \lim_n \phi(X_n) = \inf_n \max_{\mu \in \Lambda} \crl{\ang{X_n,\mu} - \phi^*_{C_V}(\mu)}\\
&=& \max_{\mu \in \Lambda} \inf_n \crl{\ang{X_n,\mu} - \phi^*_{C_V}(\mu)}
= \max_{\mu \in \Lambda} \crl{\ang{X,\mu} - \phi^*_{C_V}(\mu)}.
\eeas
In particular, $\phi(X_n) \downarrow \phi(X)$.
\end{proof}

\begin{Remark}
Assume $E$ is a Polish space and denote by $S_V$ the set of all Suslin functions 
$X \colon E \to \mathbb{R}$ generated by $C_V$ such that $X/V$ is bounded. Then 
$S_V$ equals the set of all upper semianalytic functions $X \colon E \to \mathbb{R}$ 
such that $X/V$ is bounded (see Proposition 7.41 of \cite{BS}), and every upper semianalytic 
function is measurable with respect to the universal completion of the Borel $\sigma$-algebra on $E$
(see Corollary 7.42.1 of \cite{BS}). Since every Borel measure on $E$ has a unique extension to the 
universal completion of the Borel $\sigma$-algebra, $\ang{X,\mu}$ is well-defined for all $X \in S_V$ and $\mu \in ca^+_V$.
So if $\phi : S_V \to \mathbb{R}$ is an increasing convex functional satisfying (R2) and 
$\phi(X_n) \uparrow \phi(X)$ for every sequence $(X_n)$ in $S_V$ such that 
$X_n \uparrow X$ for some $X \in S_V$, it follows exactly as in the proof of Theorem \ref{thm:dual} that 
$$
\phi(X) = \sup_{\mu \in ca^+_V} \crl{\ang{X,\mu} - \phi^*_{C_V}(\mu)} \quad \mbox{for all } X \in S_V.
$$
\end{Remark}

\setcounter{equation}{0}
\section{Proofs of the main results} 
\label{sec:proofs}

\subsection{Proof of Theorem \ref{thm1}}

For the proof of Theorem \ref{thm1} we need the following lemmas:

\begin{lemma}
\label{lem:conjugate.bounded}
If $u$ satisfies {\rm (U2)--(U3)}, then $\sup_{(\omega,y)\in\Omega\times[0,n]}|v(\omega,y)|< \infty$
for every $n\in\mathbb{N}$.
\end{lemma}

\begin{proof}
Fix $n \in \mathbb{N}$. By (U3), there exists a constant $x_0 \le 0$ such that 
$\sup_{\omega} u(\omega,x) \leq (n+1)x$ for all $x\leq x_0$. 
On the other hand, it follows from (U2) that $c:= \sup_\omega\sup_{x\geq x_0} |u(\omega,x)| \in \mathbb{R}$.
Now, let $x \in \mathbb{R}$ and $y \in [0,n]$. Then 
$$
\begin{array}{ll}
u(\omega,x)- xy \leq (n+1)x -xn \leq 0 & \mbox{ if } x \le x_0\\
u(\omega,x)- xy\leq c- x_0n & \mbox{ if } x \ge x_0.
\end{array}
$$
This shows that $v(\omega,y)= \sup_{x\in\mathbb{R}} (u(\omega,x)-xy) \leq c- x_0n$.
On the other hand, $v(\omega,y)\geq u(\omega,0)\geq -c$, and the proof is complete.
\end{proof}

\begin{lemma}
\label{lem:mart}
If $u$ satisfies {\rm (U1)--(U2)}, then there exists a constant $c \in \mathbb{R}$ such that 
\[
q \mathbb{E}^{\q} Y^- \le q \mathbb{E}^{\q} X^+ - \mathbb{E}^{\p} u(X+Y) + \mathbb{E}^{\p} v \brak{q \frac{d\q}{d\p}} + c
\]
for all Borel measurable functions $X,Y \colon \Omega \to \mathbb{R}$, every $q \in \mathbb{R}_+$ and every pair of
Borel probability measures $\p$ and $\q$ on $\Omega$ such that $q d\q \ll d\p$.
\end{lemma}

\begin{proof}
Since $u$ satisfies (U1)--(U2), $c :=  \sup_{(\omega, x) \in \Omega \times \mathbb{R}_+} \crl{u(\omega, x) - u(\omega,0)}$ 
is finite and satisfies
\[
\mathbb{E}^{\p} u(X+Y) \le \mathbb{E}^{\p} u \brak{- (X+Y)^-} + c.
\]
Moreover, it follows from the definition of $v$ that 
\[
 (X+Y)^- q \frac{d\q}{d\p} \le - u \brak{- (X+Y)^-} + v \brak{q \frac{d\q}{d\p}}.
\] 
Hence,
\beas
q \mathbb{E}^{\q} Y^- &\le& q \mathbb{E}^{\q} X^+ + q \mathbb{E}^{\q} (X+Y)^- \le q \mathbb{E}^{\q} X^+
- \mathbb{E}^{\p} u \brak{- (X+Y)^-} + \mathbb{E}^{\p} v \brak{q\frac{d\q}{d\p}}\\ &\le& 
q \mathbb{E}^{\q} X^+ - \mathbb{E}^{\p} u \brak{X+Y} + \mathbb{E}^{\p} v \brak{q \frac{d\q}{d\p}} + c.
\eeas
\end{proof}

\begin{lemma}
\label{lem:weak.dual}
If $u$ satisfies {\rm (U1)--(U3)}, then the functional
\[D(X):=\inf_{(q,\mathbb{Q})\in \hat{\cal Q}_Z} \crl{q\mathbb{E}^\mathbb{Q}X + D^{\alpha}_v(q\q)}\]
satisfies $U(X) \leq D(X) < \infty$ for all $X\in B_Z$.
\end{lemma}

\begin{proof}
By Lemma \ref{lem:conjugate.bounded}, one has 
$D(X) \le D^{\alpha}_v(0) = \inf_{\p \in {\cal P}} \crl{\mathbb{E}^{\p} v(0) + \alpha(\p)} < \infty$ for all $X \in B_Z$.

Now, consider $\p \in {\cal P}$, $X\in B_Z$ and $\vartheta\in\Theta$ such that 
$\mathbb{E}^{\mathbb{P}} u \brak{X+ \sum_{t=1}^T \vartheta_t \Delta S_t} > - \infty$.
It is immediate from the definition of $v$ that
\[
\mathbb{E}^{\mathbb{P}} u \brak{X+ \sum_{t=1}^T \vartheta_t \Delta S_t} \le \mathbb{E}^{\p} v(0).
\]
Moreover, for $q \in (0, \infty)$ and $\mathbb{Q}\in\mathcal{Q}_Z$ such that $q\q \ll \p$ and
$\mathbb{E}^{\q} v(q d\q/d\p) < \infty$, one  obtains from Lemma \ref{lem:mart} that there exists a constant $c \in \mathbb{R}$ 
such that
\[
q \mathbb{E}^{\q} \brak{\sum_{t=1}^T \vartheta_t \Delta S_t}^- 
\le q \mathbb{E}^{\q} X^+ - \mathbb{E}^{\mathbb{P}} u \brak{X+ \sum_{t=1}^T \vartheta_t \Delta S_t}
+ \mathbb{E}^{\mathbb{P}}v \brak{q\frac{d\mathbb{Q}}{d\mathbb{P}}} + c
< \infty.
\]
So it follows from Theorems 1 and 2 in \cite{jacod1998local} that 
$\brak{\sum_{s=1}^t \vartheta^n_s \Delta S_s}_{t=0}^T$ is a $\mathbb{Q}$-martingale, and therefore,
$\mathbb{E}^{\q} \sum_{t=1}^T \vartheta_t \Delta S_t =0$. By the definition of $v$, one has
\[
u\brak{X+ \sum_{t=1}^T \vartheta_t \Delta S_t} \le
\brak{X+ \sum_{t=1}^T \vartheta_t \Delta S_t} q \frac{d\q}{d\p} + v \brak{q\frac{d\q}{d\p}}.
\]
Hence,
\[
\mathbb{E}^{\mathbb{P}} u\brak{X+ \sum_{t=1}^T \vartheta_t \Delta S_t} 
  \leq q \mathbb{E}^{\mathbb{Q}} \edg{X + \sum_{t=1}^T \vartheta_t \Delta S_t}
    +\mathbb{E}^{\mathbb{P}}v\Big(q\frac{d\mathbb{Q}}{d\mathbb{P}}\Big) 
    =q \mathbb{E}^{\mathbb{Q}}X+\mathbb{E}^{\mathbb{P}}v\Big(q\frac{d\mathbb{Q}}{d\mathbb{P}}\Big).
\]
Now, first taking the infimum in
\[
\mathbb{E}^{\mathbb{P}} u\brak{X+ \sum_{t=1}^T \vartheta_t \Delta S_t} + \alpha(\p) 
\le q \mathbb{E}^{\mathbb{Q}}X+\mathbb{E}^{\mathbb{P}}v\Big(q\frac{d\mathbb{Q}}{d\mathbb{P}}\Big)
+ \alpha(\p)
\]
over all $\mathbb{P}\in\mathcal{P}$ and $(q,\q) \in\hat{\cal Q}_Z$ such that $q \q \ll \p$ 
and then the supremum over all $\vartheta \in \Theta$, yields $U(X) \le D(X)$.
\end{proof}

\begin{lemma} \label{lemma:Dv}
If $u$ satisfies {\rm (U1)--(U2)}, then
\[
D_v(q\q \,\| \, \p) = \sup_{X \in C_Z} \crl{\mathbb{E}^{\p} u(X) - q \mathbb{E}^{\q} X}
\]
for all $q \in \mathbb{R}_+$ and $\p, \q \in {\cal M}_Z$.  
\end{lemma}

\begin{proof}
First, note that if $q \in \mathbb{R}_+$ and $\p,\q \in {\cal M}_Z$ are such that
$q \q$ is not absolutely continuous with respect to $\p$, there exists a 
Borel set $A \subseteq \Omega$ such that $q \q[A] > 0$ and $\p[A] = 0$. Since 
$\q$ is a regular measure, there is a closed set $K \subseteq A$ such that $q \q[K] > 0$ and 
$\p[K] =0$. For every $m \in \mathbb{N}$, there exists a sequence of bounded continuous functions 
$X_n \colon \Omega \to \mathbb{R}$ such that $X_n \downarrow m 1_K$. It follows 
from the monotone convergence theorem that
\[
\mathbb{E}^{\p} u(- X_n) + q \mathbb{E}^{\q} X_n \to
 \mathbb{E}^{\p} u(-m1_K) + q \mathbb{E}^{\q}[m1_K] = \mathbb{E}^{\p} u(0) + q m \q[K],
\]
and as a consequence,
\[
\sup_{X \in C_Z} \crl{\mathbb{E}^{\p} u(X) - q \mathbb{E}^{\q} X} = \infty = D_v(q\q \,\| \, \p).
\]

Next, assume that $q \q$ is absolutely continuous with respect to $\p$. Then,
\[
\mathbb{E}^{\p} u(X) - q \mathbb{E}^{\q} X = \mathbb{E}^{\p} \edg{u(X) - q \frac{d\q}{d\p} X} \le
\mathbb{E}^{\p} v \brak{q \frac{d\q}{d\p}}
\]
for all $X \in C_Z$. On the other hand, there exists a 
sequence of simple random variables $(Y_n)$ such that 
\[
\mathbb{E}^{\p} \edg{u(Y_n) - q \frac{d\q}{d\p} Y_n} \to
\mathbb{E}^{\p} v \brak{q \frac{d\q}{d\p}},
\]
from which it follows that there exists a sequence $(X_n)$ in $C_Z$ such that
\[
\mathbb{E}^{\p} \edg{u(X_n) - q \frac{d\q}{d\p} X_n} \to
\mathbb{E}^{\p} v \brak{q \frac{d\q}{d\p}} = D_v(q\q \,\| \, \p) .
\]
This completes the proof of the lemma.
\end{proof}

\begin{lemma} \label{lemma:inf}
Assume {\rm (U1)--(U3)} and {\rm (A2)} hold. Then, for every constant $m \in \mathbb{R}_+$, there exists
a $c \in \mathbb{R}_+$ such that 
\[
\inf_{\p \in {\cal P}} \crl{\mathbb{E}^{\p} u(X) + \alpha(\p)} = \inf_{\p \in {\cal P}_c} \crl{\mathbb{E}^{\p} u(X) + \alpha(\p)}
\]
for all Borel measurable functions $X \colon \Omega \to \mathbb{R}$ satisfying $X \ge - mZ$.\end{lemma}

\begin{proof}
Fix $m \in \mathbb{R}_+$. It follows from (U2)--(U3) and (A2) that
\[
\varphi(x) := \inf_{\p \in {\cal P}} \crl{x \mathbb{E}^{\p}u(-mZ) + \alpha(\p)}
\]
is finite for all $x \in \mathbb{R}_+$. So, the function $\psi \colon \mathbb{R} \to (\infty,\infty]$, given by 
\[
\psi(y) := \sup_{x \ge 0} \crl{xy + \varphi(x)},
\]
is increasing and satisfies $\lim_{y \to \infty} \psi(y)/y \to \infty$. As a consequence, the right-continuous inverse 
\[
\psi^{-1}(y) := \inf \crl{x \in \mathbb{R} : \psi(x) > y}
\]
has the property $\lim_{y \to \infty} \psi^{-1}(y)/y = 0$. Since
\[
\alpha(\p) \ge \varphi(x) - x \mathbb{E}^{\p}u(-mZ) 
\]
for all $x \in \mathbb{R}_+$, one has
\[
\alpha(\p) \ge \psi(-\mathbb{E}^{\p}u(-mZ)),
\]
and therefore,
\[
\mathbb{E}^{\p} u(-mZ) \ge - \psi^{-1}(\alpha(\p)).
\]
By (U1), one has for all $X \ge - mZ$, 
\[
\inf_{\p \in {\cal P}} \crl{\mathbb{E}^{\p}u(X) + \alpha(\p)} \le \mathbb{E}^{\p} u(\infty) < \infty
\]
and 
\[
\mathbb{E}^{\p} u(X) + \alpha(\p) \ge \mathbb{E}^{\p} u(-mZ) + \alpha(\p) \ge - \psi^{-1}(\alpha(\p)) + \alpha(\p).
\]
Since $\lim_{c \to \infty} c -\psi^{-1}(c) = \infty$, this 
shows that there exists a $c \in \mathbb{R}$ such that 
\[
\inf_{\p \in {\cal P}} \crl{\mathbb{E}^{\p}u(X) + \alpha(\p)} = 
\inf_{\p \in {\cal P}_c} \crl{\mathbb{E}^{\p}u(X) + \alpha(\p)} 
\]
for all $X \ge -mZ$.
\end{proof} 

Next, note that if $u$ satisfies (U2), then for every continuous function $\gamma \colon [1,\infty) \to \mathbb{R}$,
\[
Z_{\gamma} := 1 \vee (-u(- \gamma(Z)))
\]
defines a continuous function from $\Omega$ to $[1,\infty)$.

\begin{lemma} \label{lemma:Pc}
Assume {\rm (U2)--(U3)} and {\rm (A1)--(A2)} hold. Then, there exists a 
continuous increasing function $\gamma \colon [0,\infty) \to \mathbb{R}$ such that
$\lim_{x \to \infty} \gamma(x)/x = \infty$, and for all $c \in \mathbb{R}_+$, ${\cal P}_c$ 
is a $\sigma({\cal M}_{Z_{\gamma}}, C_{Z_{\gamma}})$-compact 
subset of ${\cal M}_{Z_{\gamma}}$.
\end{lemma}

\begin{proof}
By (A2), there exists an increasing function $\beta \colon [1,\infty) \to \mathbb{R}$ 
such that $\lim_{x \to \infty} \beta(x)/x = \infty$ and  
$\inf_{\mathbb{P} \in \mathcal{P}} \crl{\mathbb{E}^\mathbb{P}u(-\beta(Z)) + \alpha(\p)} > - \infty$. 
So one can construct a continuous increasing function $\gamma \colon [1,\infty) \to \mathbb{R}$
such that $\lim_{x\to \infty} \gamma(x)/x = \lim_{x \to \infty} \beta(x)/\gamma(x) = \infty$.
It follows from (U3) that there exists a $z_0 \in \mathbb{R}$ such that 
$u(-\gamma(Z)) \le - Z$ on $\crl{Z > z_0}$. This shows that $C_Z \subseteq C_{Z_{\gamma}}$ and 
${\cal M}_{Z_{\gamma}} \subseteq {\cal M}_Z$. Since for given $c \in \mathbb{R}_+$, one has
\[
\inf_{\p \in {\cal P}_c} \mathbb{E}^{\p} u(- \beta(Z)) + c
\ge \inf_{\p \in {\cal P}_c} \crl{\mathbb{E}^{\p} u(- \beta(Z)) + \alpha(\p)} > - \infty,
\]
one obtains
\[
\lim_{z \to \infty} \sup_{\p \in {\cal P}_c} \mathbb{E}^{\p} [Z_{\gamma} 1_{\crl{Z > z}}] = 0.
\]
Moreover, it follows from (U2) that $Z_{\gamma}$ is bounded on the sets $\crl{Z \le z}$. Hence, 
${\cal P}_c$ is contained in ${\cal M}_{Z_{\gamma}}$, and since by (A1), it is 
$\sigma({\cal M}_Z,C_Z)$-closed, it is also $\sigma({\cal M}_{Z_{\gamma}},C_{Z_{\gamma}})$-closed.
Note that $\p \mapsto Z_{\gamma} d\p$ transforms ${\cal P}_c$ into a subset $\tilde{\cal P}_c$ of the 
finite Borel measures ${\cal M}$ on $\Omega$. Since the sets $\crl{Z \le z}$ are compact,
it follows from Prokhorov's theorem that $\tilde{\cal P}_c$ is $\sigma({\cal M}, C_b)$-compact, 
where $C_b$ are all bounded continuous functions on $\Omega$. 
But this is equivalent to ${\cal P}_c$ being $\sigma({\cal M}_{Z_{\gamma}}, C_{Z_{\gamma}})$-compact.
\end{proof}

Next, let us denote by $\tilde{\Theta}$ the set of all strategies $\vartheta \in \Theta$ such that 
$\vartheta_t$ is continuous and bounded for all $t=1,\dots,T$, and define
$$
\tilde{U}(X) := \sup_{\vartheta\in\tilde{\Theta}} \inf_{\mathbb{P}\in\mathcal{P}} 
\crl{\mathbb{E}^\mathbb{P}u \brak{X+ \sum_{t=1}^T \vartheta_t \Delta S_t} + \alpha(\p)}, \quad X \in B_Z,
$$
as well as
$$
\tilde{U}^*_{C_Z}(q \q) := \sup_{X \in C_Z} \crl{\tilde{U}(X) - q \mathbb{E}^{\q} X}, \quad 
\mbox{for } q \in \mathbb{R}_+ \mbox{ and } \q \in {\cal M}_Z.
$$
Then the following holds:

\begin{lemma} \label{lemma:Ut}
If {\rm (U1)--(U3)} and {\rm (A1)--(A2)} hold, then $\tilde{U}$ is an increasing concave mapping from 
$B_Z$ to $\mathbb{R}$ satisfying
\be \label{upL}
\tilde{U}(X_n) \uparrow \tilde{U}(X) \mbox{ for every sequence $(X_n)$ in $C_Z$ such that $X_n\uparrow X$
for some $X \in L_Z$}
\ee
and 
\be \label{U*}
\tilde{U}^*_{C_Z} (q\q) = 
\begin{cases}
D^{\alpha}_v(q\q) & \mbox{ if } q= 0 \mbox{ or }\q \in {\cal Q}_Z\\
\infty & \mbox{ if } q > 0 \mbox{ and } \q \in {\cal M}_Z \setminus {\cal Q}_Z.
\end{cases}
\ee
\end{lemma}

\begin{proof}
It is straight-forward to check that $\tilde{U}$ is an increasing concave mapping from $B_Z$ to $\mathbb{R}$ .
To show \eqref{U*}, we note that for given $q \in \mathbb{R}_+$ and $\q \in {\cal M}_Z$,
\begin{align*} 
\tilde{U}^*_{C_Z}(q\mathbb{Q})
&= \sup_{X\in C_Z} \sup_{\vartheta\in\tilde{\Theta}} \inf_{\mathbb{P}\in\mathcal{P}}
\crl{\mathbb{E}^\mathbb{P} u\brak{X + \sum_{t=1}^T \vartheta_t \Delta S_t} + \alpha(\p)
- q\mathbb{E}^\mathbb{Q} X}\nonumber\\  
&= \sup_{X\in C_Z} \sup_{\vartheta\in\tilde{\Theta}}\inf_{\mathbb{P}\in\mathcal{P}}
\crl{ \mathbb{E}^{\p} u(X) + \alpha(\p) - q\mathbb{E}^\mathbb{Q} X + q\mathbb{E}^\mathbb{Q} 
\sum_{t=1}^T \vartheta_t \Delta S_t}.
\end{align*}
Since $\mathbb{E}^\mathbb{Q} \sum_{t=1}^T \vartheta_t \Delta S_t = 0$ for all $\vartheta \in \tilde{\Theta}$ 
if and only if $S$ is a $\q$-martingale, one has $\tilde{U}^*_{C_Z}(q\mathbb{Q}) = \infty$ 
for $q > 0$ and $\q \in {\cal M}_Z \setminus {\cal Q}_Z$. On the other hand, if $q = 0$ or $\q \in {\cal Q}_Z$,
then 
\[
\tilde{U}^*_{C_Z}(q\mathbb{Q}) = \sup_{X \in C_Z} \inf_{\mathbb{P}\in\mathcal{P}}
\crl{ \mathbb{E}^{\p} u(X) + \alpha(\p) - q\mathbb{E}^\mathbb{Q} X}.
\]
So, it follows from Lemma \ref{lemma:Dv} that
\be \label{U*le} \tilde{U}^*_{C_Z}(q\mathbb{Q}) \le
\inf_{\mathbb{P}\in\mathcal{P}} \sup_{X \in C_Z} \crl{ \mathbb{E}^{\p} u(X) 
+ \alpha(\p) - q\mathbb{E}^\mathbb{Q} X}
= \inf_{\mathbb{P} \in\mathcal{P}} \crl{D_v(q\q \,\| \, \p) + \alpha(\p)}.
\ee
By Lemma \ref{lemma:Pc}, there exists a continuous increasing function 
$\gamma \colon [1,\infty) \to \mathbb{R}$ such that $\lim_{x \to \infty} \gamma(x)/x = \infty$,
and for all $c \in \mathbb{R}_+$, ${\cal P}_c$ is a 
$\sigma({\cal M}_{Z_{\gamma}}, C_{Z_{\gamma}})$-compact subset of ${\cal M}_{Z_{\gamma}}$.
For a given constant $m \in \mathbb{R}_+$, denote 
\[
C^m_Z := \crl{X \in C_Z : X \ge - mZ} \quad \mbox{and} \quad
D^m_v(q \q \,\| \, \p) := \sup_{X \in C^m_Z} \crl{ \mathbb{E}^{\p} u(X) - q\mathbb{E}^\mathbb{Q} X}.
\]
By Lemma \ref{lemma:inf}, there exists an $a \in \mathbb{R}_+$ such that 
\[
\sup_{X \in C^m_Z} \inf_{\p \in {\cal P}} \crl{ \mathbb{E}^{\p} u(X) + \alpha(\p) - q\mathbb{E}^\mathbb{Q} X}
= \sup_{X \in C^m_Z} \inf_{\p \in {\cal P}_a} \crl{ \mathbb{E}^{\p} u(X) + \alpha(\p) - q\mathbb{E}^\mathbb{Q} X}.
\]
So, since $\mathbb{E}^{\p} u(X) + \alpha(\p) - q\mathbb{E}^\mathbb{Q} X$ is concave 
in $X \in C^m_Z$ as well as convex and $\sigma({\cal M}_{Z_{\gamma}}, C_{Z_{\gamma}})$-lower 
semicontinuous in $\p \in {\cal P}_c$,
it follows from the minimax result, Theorem 2 of \cite{KF}, that
\beas
\sup_{X \in C^m_Z} \inf_{\p \in {\cal P}} \crl{ \mathbb{E}^{\p} u(X) + \alpha(\p) - q\mathbb{E}^\mathbb{Q} X}
&=& \inf_{\p \in {\cal P}_a} \sup_{X \in C^m_Z} \crl{ \mathbb{E}^{\p} u(X) + \alpha(\p) - q\mathbb{E}^\mathbb{Q} X}\\
&\ge& \inf_{\p \in {\cal P}} \crl{D^m_v(q \q \,\| \, \p) + \alpha(\p)}.
\eeas
Now, note that 
\[
- \infty < \inf_{\omega \in \Omega} u(\omega,0) \le
D^m_v(q \q \,\| \, \p) \le \sup_{(\omega,x) \in \Omega \times \mathbb{R}_+}
u(\omega, x) + q m \mathbb{E}^{\q} Z < \infty \quad \mbox{for all } \p \in {\cal P}.
\]
Moreover, $D^m_v(q \q \,\| \, \p) + \alpha(\p)$ is increasing in $m \in \mathbb{R}_+$ as well as convex and 
$\sigma({\cal M}_{Z_{\gamma}}, C_{Z_{\gamma}})$-lower semicontinuous in $\p \in {\cal P}$. 
So, if there exists a $b \in \mathbb{R}_+$ such that 
\be \label{bB}
\sup_{m \in \mathbb{R}_+} \inf_{\p \in {\cal P}} \crl{D^m_v(q \q \,\| \, \p) + \alpha(\p)} 
= \sup_{m \in \mathbb{R}_+} \inf_{\p \in {\cal P}_{b}} \crl{D^m_v(q \q \,\| \, \p) + \alpha(\p)},
\ee
another application of Theorem 2 in \cite{KF} yields
\beas
\sup_{m \in \mathbb{R}_+} \inf_{\p \in {\cal P}} \crl{D^m_v(q \q \,\| \, \p) + \alpha(\p)}
&=& \inf_{\p \in {\cal P}_b} \sup_{m \in \mathbb{R}_+} \crl{D^m_v(q \q \,\| \, \p) + \alpha(\p)}\\
&\ge& \inf_{\p \in {\cal P}} \crl{D_v(q \q \,\| \, \p) + \alpha(\p)}.
\eeas
On the other hand, if \eqref{bB} does not hold for any $b \in \mathbb{R}_+$, there exists a 
sequence $(b_n)$ in $\mathbb{R}_+$ such that $b_n \to \infty$ and
\[
\sup_{m \in \mathbb{R}_+} \inf_{\p \in {\cal P}} \crl{D^m_v(q \q \,\| \, \p) + \alpha(\p)}
\ge \lim_{n \to \infty} \crl{\inf_{\omega \in \Omega} u(\omega,0) + b_n} = \infty.
\]
This shows that
\[
\tilde{U}^*_{C_Z}(q\mathbb{Q}) = \sup_{m \in \mathbb{R}_+} 
\sup_{X \in C^m_Z} \inf_{\mathbb{P}\in\mathcal{P}} 
\crl{ \mathbb{E}^{\p} u(X) + \alpha(\p) - q\mathbb{E}^\mathbb{Q} X}
\ge \inf_{\p \in {\cal P}} \crl{D_v(q \q \,\| \, \p) + \alpha(\p)},
\]
which, together with \eqref{U*le}, implies \eqref{U*}.

Next, consider a sequence $(X_n)$ in $C_Z$ such that $X_n \uparrow X$ for some $X \in C_Z$.
Since $X_1 \in C_Z$, one has $X_1 \ge - mZ$ for some $m \in \mathbb{R}_+$. So, by Lemma 
\ref{lemma:inf}, there exists a $c \in \mathbb{R}_+$ such that
\[
\inf_{\p \in {\cal P}} \crl{\mathbb{E}^{\p} u(X_n) + \alpha(\p)}
= \inf_{\p \in {\cal P}_c} \crl{\mathbb{E}^{\p} u(X_n) + \alpha(\p)} \quad \mbox{for all } n.
\]
Using Theorem 2 of \cite{KF} once more, we obtain
\[
\sup_n \inf_{\p \in {\cal P}} \crl{\mathbb{E}^{\p} u(X_n) + \alpha(\p)} 
= \inf_{\p \in {\cal P}_c} \sup_n \crl{\mathbb{E}^{\p} u(X_n) + \alpha(\p)} 
\ge \inf_{\p \in {\cal P}} \crl{\mathbb{E}^{\p} u(X) + \alpha(\p)},
\]
which by monotonicity, gives
\be \label{XnX}
\sup_n \inf_{\p \in {\cal P}} \crl{\mathbb{E}^{\p} u(X_n) + \alpha(\p)} 
= \inf_{\p \in {\cal P}} \crl{\mathbb{E}^{\p} u(X) + \alpha(\p)}.
\ee
Since, for a given strategy $\vartheta \in \tilde{\Theta}$, $\sum_{t=1}^T \vartheta_t \Delta S_t$ belongs to $C_Z$, 
we get from \eqref{XnX} that
$$
\sup_n \inf_{\p \in {\cal P}} \mathbb{E}^{\p} 
u \brak{X_n + \sum_{t=1}^T \vartheta_t \Delta S_t} = 
\inf_{\p \in {\cal P}} \mathbb{E}^{\p} u \brak{X + \sum_{t=1}^T \vartheta_t \Delta S_t},
$$
which, due to $\tilde{U}(X) = \sup_{\vartheta \in \tilde{\Theta}} 
\inf_{\p \in {\cal P}} \mathbb{E}^{\p} u \brak{X + \sum_{t=1}^T \vartheta_t \Delta S_t}$, implies
that $\tilde{U}$ satisfies \eqref{upL} for $X \in C_Z$. In particular, 
$\phi(X) = - \tilde{U}(-X)$ is an increasing convex mapping from $B_Z$ to $\mathbb{R}$ 
satisfying condition (R1) of Theorem \ref{thm:dual}. Moreover,
\beas
\phi^*_{C_Z}(q\q) &=& \tilde{U}^*_{C_Z}(q\q) \le \phi^*_{U_Z}(q\q) = \tilde{U}^*_{L_Z}(q\q) :=
\sup_{X \in L_Z} \inf_{\mathbb{P}\in\mathcal{P}} \crl{ \mathbb{E}^{\p} u(X) + \alpha(\p) - q\mathbb{E}^\mathbb{Q} X}\\
&\le& \inf_{\mathbb{P} \in\mathcal{P}} \crl{D_v(q\q \,\| \, \p) + \alpha(\p)} = \tilde{U}^*_{C_Z}(q\q)
\eeas
for all $(q,\q) \in \hat{\cal Q}_Z$.
So it follows from Proposition \ref{prop:dual} that $\phi$ satisfies condition (R2) of Theorem \ref{thm:dual}, 
which means that $\tilde{U}$ satisfies \eqref{upL}.
\end{proof}

Now, we are ready to prove Theorem \ref{thm1}. \medskip

\noindent
{\bf Proof of Theorem \ref{thm1}}\\
It follows from Lemma \ref{lemma:Ut} and Theorem \ref{thm:dual} that 
\[
\tilde{U}(X) = \min_{(q,\q) \in \hat{\cal Q}_Z} \crl{q \mathbb{E}^{\q} X + D^{\alpha}_v(q\q)} \quad \mbox{for all } X \in L_Z.
\]
Since, by Lemma \ref{lem:weak.dual}, 
\[
\tilde{U}(X) \le U(X) \le \inf_{(q,\q) \in \hat{\cal Q}_Z} \crl{q \mathbb{E}^{\q} X + D^{\alpha}_v(q\q)} \quad \mbox{for all } X \in L_Z,
\]
this proves the theorem. \qed

\subsection{Medial limits}

To prove Theorem \ref{thm2} and Corollary \ref{cor}, we need the concept 
of a medial limit, which for our purposes, is a positive linear functional, $\lm \colon l^{\infty} \to \R$, satisfying
$\liminf\leq\lm\leq\limsup$ such that for any uniformly bounded sequence $X_n\colon M \to \mathbb{R}$
of universally measurable functions on a measurable space $(M,{\cal F})$, $X = \lm_n X_n$ 
is universally measurable and $\mathbb{E}^{\mathbb{P}} X = \lm_n \mathbb{E}^{\mathbb{P}} X_n$ 
for every probability measure $\mathbb{P}$ on the universal completion ${\cal F}^*$ of ${\cal F}$.
It was originally shown by Mokobozki that medial limits exist under the usual ZFC axioms and the
continuum hypothesis; see \cite{Meyer}. Later, Normann \cite{Normann} showed that it is enough to 
assume ZFC and Martin's axiom. If a medial limit exists, we extend it to $\overline{\mathbb{R}}^{\mathbb{N}}$ by setting
\be \label{ext}
\lm_n x_n := \sup_{k \in \mathbb{N}} \inf_{m \in \mathbb{N}} \lm_n \, (-m) \vee (x_n \wedge k).
\ee

\begin{lemma} \label{lem:mllemma} 
Assume a medial limit exists. Then the following hold:
\begin{itemize}
\item[{\rm (i)}]
The set ${\cal L}$ of sequences $(x_n)$ in $\mathbb{R}^{\mathbb{N}}$ satisfying
$\lm_n |x_n| < \infty$ is a linear space.
\item[{\rm (ii)}]
$\lm \colon {\cal L} \to \mathbb{R}$ is a positive linear functional.
\item[{\rm (iii)}]
$\varphi(\mathop{\mathrm{lim\,med}}_n x_n) \leq\mathop{\mathrm{lim\,med}}_n \varphi(x_n)$
for every convex function $\varphi \colon\mathbb{R}\to\mathbb{R}$ 
and $(x_n) \in \mathcal{L}$.
\item[{\rm (iv)}]
$\lm_n X_n$ is universally measurable for every sequence of universally measurable functions 
$X_n \colon \Omega \to \mathbb{R}$.
\item[{\rm (v)}] 
$\mathbb{E}^\mathbb{P}\lm_n X_n  \le \mathop{\mathrm{lim\,med}}_n \mathbb{E}^\mathbb{P}X_n$
for each probability measure $\mathbb{P}$ on ${\cal F}^\ast$ and every 
sequence of universally measurable functions $X_n \colon \Omega \to \mathbb{R}$ such that 
$X_n \ge c$ for all $n$ and a constant $c \in \mathbb{R}$.
\end{itemize}
\end{lemma}

\begin{proof}
(i) and (ii) are simple consequences of \eqref{ext}. To show (iii), we note that by the Fenchel--Moreaux 
theorem, $\varphi$ can be written as $\varphi(x) = \sup_{y \in \mathbb{R}} xy - \varphi^*(y)$
for the convex conjugate $\varphi^*$ of $\varphi$. Moreover, since $\liminf \le \lm \le \limsup$,
one has $\lm_n(x_n) = c$ for constant sequences $x_n \equiv c$. So, since $\lm$ is linear on ${\cal L}$, 
one obtains
$$
\varphi(\mathop{\mathrm{lim\,med}}_n x_n)
= \sup_{y \in \mathbb{R}}  
\big(\lm_n x_n y -\varphi^*(y)\big)
\leq \lm_n \big(\sup_{y \in \mathbb{R}} x_ny -\varphi^*(y)\big)
=\lm_n \varphi(x_n).
$$  

(iv) follows from \eqref{ext} since $\lm_n X_n$ is universally measurable 
for any uniformly bounded sequence of universally measurable functions 
$X_n \colon \Omega \to \mathbb{R}$.

(v): For every $k \in \mathbb{N}$, 
$$
\mathbb{E}^{\mathbb{P}} \lm_n (X_n \wedge k) = \lm_n \mathbb{E}^{\mathbb{P}}
(X_n \wedge k) \le \lm_n \mathbb{E}^{\mathbb{P}} X_n,
$$
and therefore, by \eqref{ext} and the monotone convergence theorem,
$\mathbb{E}^{\mathbb{P}} \lm_n X_n \le \lm_n \mathbb{E}^{\mathbb{P}} X_n.$
\end{proof}

\subsection{Proofs of Theorem \ref{thm2} and Corollary \ref{cor}}

\begin{lemma}\label{lem:existence.strategy}
Assume a medial limit exists, $u$ fulfills {\rm (U1)--(U3)} and ${\cal P}$ satisfies {\rm (NA)}. 
Let $X_n \colon \Omega \to \mathbb{R}$ be a sequence of Borel measurable functions 
decreasing pointwise to a Borel measurable function $X \colon \Omega \to \mathbb{R}$ such that $U(X) \in \mathbb{R}$.
Then $U(X_n)$ decreases to $U(X)$, and there exists a strategy $\vartheta^\ast \in \Theta$ such that 
\[ U(X)=\inf_{\mathbb{P}\in\mathcal{P}} \crl{
\mathbb{E}^\mathbb{P} u \brak{X+ \sum_{t=1}^T \vartheta^\ast_t \Delta S_t} + \alpha(\p)}. 
\] 
\end{lemma}

\begin{proof}
Since $U$ is bounded from above, there exists for each $n$, a $\vartheta^n\in\Theta$ such that
\begin{align*}
  \inf_{\mathbb{P}\in\mathcal{P}} 
  \crl{\mathbb{E}^\mathbb{P}u \brak{X_n + \sum_{t=1}^T \vartheta^n_t \Delta S_t} + \alpha(\p)}
  \geq U(X_n)- \frac{1}{n}.
  \end{align*}
Denote $A_t^{\pm} :=\{\omega \in \Omega:\mathop{\mathrm{lim\,med}}_n 
(\vartheta_t^n(\omega))^{\pm} = \infty\}$ 
and define
  \[ \vartheta^\ast_t(\omega):=
  \begin{cases}
  \lm_n \vartheta^n_t(\omega) & \text{if } \omega\notin A_t^+ \cup A_t^-\\
  0&\text{otherwise.}
  \end{cases}\] 
We want to show that
\begin{align}
  \label{eq:varthat.in.L}
  \p\edg{ \lm_n |\vartheta^n_t \Delta S_t|< \infty } =1
  \quad\text{ for all $t = 1, \dots, T$ and }\p\in\mathcal{P}.
  \end{align}
To do that, we note that by (NA), every $\mathbb{P} \in {\cal P}$ is dominated by a $\p' \in {\cal P}$ that does not 
admit arbitrage. By the fundamental theorem of asset pricing, there exists
a martingale measure $\mathbb{Q}$ equivalent to $\mathbb{P}'$ such that $\mathbb{E}^{\q} X^+_1 < \infty$ and
$d\mathbb{Q}/d\mathbb{P}'$ is bounded\footnote{To see this, note that 
$d\tilde{\p}/d\p' = (1/1+ X^+_1)/\mathbb{E}^{\p'} (1/1+ X^+_1)$ defines a measure $\tilde{\p}$ equivalent 
to $\p'$ such that $\mathbb{E}^{\tilde{\p}} X^+_1 < \infty$. 
$\tilde{\p}$ still does not admit arbitrage. Therefore, there exists a martingale measure $\q$ with bounded 
density $d\q/d\tilde{\p}$; see e.g. Theorem 5.17 in \cite{FS}. $\q$ is equivalent to $\p'$ such that 
$\mathbb{E}^{\q} X^+_1 < \infty$ and $d\q/d\p'$ is bounded.
}.
If we can show that
\be \label{qinfty}
\lm_n |\vartheta^n_t \Delta S_t|< \infty \quad \mbox{$\q$-almost surely}
\ee
for all $t = 1, \dots, T$, \eqref{eq:varthat.in.L} follows since $\q$ dominates $\p$.
To prove \eqref{qinfty}, we set $\vartheta^n_0 = 0$ and use an induction argument.
Fix $t \ge 1$, and assume that \eqref{qinfty} holds for all $s \le t-1$. 

Since 
\[
\mathbb{E}^{\p'} u\brak{X_n + \sum_{t=1}^T \vartheta^n_t \Delta S_t} + \alpha(\p')
\ge U(X_n) \ge U(X) \in \mathbb{R},
\]
one obtains from Lemmas \ref{lem:conjugate.bounded} and \ref{lem:mart}
that there exist constants $c,c' \in \mathbb{R}$ such that 
\beas
\mathbb{E}^\mathbb{Q} \edg{\brak{\sum_{t=1}^T \vartheta^n_t \Delta S_t}^-}
&\le& \mathbb{E}^{\q} X_1^+ - \mathbb{E}^{\p'} u\brak{X_n + \sum_{t=1}^T \vartheta^n_t \Delta S_t} 
+ \mathbb{E}^{\p'} v \brak{q \frac{d\q}{d\p'}} + c\\
&\le& \mathbb{E}^{\q} X_1^+ + \alpha(\p') - U(X)
+ \mathbb{E}^{\p'} v \brak{q \frac{d\q}{d\p'}} + c = c'
\eeas
for all $n$. So it follows from Theorems 1 and 2 in \cite{jacod1998local} that 
 $\sum_{s=1}^t \vartheta^n_s \Delta S_s$ is a $\mathbb{Q}$-martingale. Consequently,
 $\brak{\sum_{s=1}^t \vartheta^n_s \Delta S_s}^-$ is a $\mathbb{Q}$-submartingale, and therefore,
$$
\mathbb{E}^\mathbb{Q}\edg{\brak{\sum_{s=1}^t \vartheta^n_s \Delta S_s}^-}
\le \mathbb{E}^\mathbb{Q}\edg{\brak{\sum_{s=1}^T \vartheta^n_s \Delta S_s}^-} \leq c'.
$$
Now, we obtain from part (v) of Lemma \ref{lem:mllemma} that $\lm_n \brak{\sum_{s=1}^t \vartheta^n_s \Delta S_s}^-$
is $\mathbb{Q}$-almost surely finite. But since 
\[
\brak{\vartheta^n_t \Delta S_t}^- \le 
\brak{\sum_{s=1}^{t-1} \vartheta^n_s \Delta S_s}^+ + \brak{\sum_{s=1}^t \vartheta^n_s \Delta S_s}^-,
\]
we get from the induction hypothesis that
$\mathop{\mathrm{lim\,med}}_n (\vartheta^n_t)^{\pm} (\Delta S_t)^{\mp}$
is $\mathbb{Q}$-almost surely finite. Since
$\lm_n (\vartheta^{n}_t)^+ (\Delta S_t)^-= \infty$ on $A^+_t \cap \{\Delta S_t< 0\}$, one has
$\mathbb{Q}[A^+_t \cap \crl{\Delta S_t < 0}]=0$. By the martingale property, this implies 
$\mathbb{Q}[A^+_t \cap \crl{\Delta S_t \neq 0}]= 0$. The same argument applied to $A_t^-$ gives 
$\mathbb{Q}[A^-_t \cap \crl{\Delta S_t \neq 0}]= 0$. 
It follows that $\lm_n|\vartheta^{n}_t \Delta S_t| < \infty$ $\q$-almost surely, which implies \eqref{eq:varthat.in.L}.

As a result, one has $\lm_n \sum_{t=1}^T \vartheta^n_t \Delta S_t = \sum_{t=1}^T \vartheta^*_t \Delta S_t$
$\p$-almost surely for all $\p\in\mathcal{P}$. Since $u$ is increasing, concave and bounded from above, 
an application of (iii) and (v) of Lemma \ref{lem:mllemma} to $-u$ gives
\beas
U(X) &\geq& \inf_{\mathbb{P}\in\mathcal{P}} \crl{
  \mathbb{E}^\mathbb{P} u\brak{X+ \sum_{t=1}^T \vartheta^*_t \Delta S_t} + \alpha(\p)}\\
  &\geq& \inf_{\mathbb{P}\in\mathcal{P}} \crl{\mathbb{E}^\mathbb{P}
  \lm_n u \brak{X_n+ \sum_{t=1}^T \vartheta^n_t \Delta S_t} + \alpha(\p)} \\
&\ge&  \inf_{\mathbb{P}\in\mathcal{P}} \crl{\mathop{\mathrm{lim\,med}}_n
  \mathbb{E}^\mathbb{P} u\brak{X_n + \sum_{t=1}^T \vartheta^n_t \Delta S_t} + \alpha(\p)}\\
  &\ge& \mathop{\mathrm{lim\,med}}_n \inf_{\mathbb{P}\in\mathcal{P}} \crl{
  \mathbb{E}^\mathbb{P} u\brak{X_n + \sum_{t=1}^T \vartheta^n_t \Delta S_t} + \alpha(\p)} =\inf_n U(X_n).
  \eeas
By monotonicity, $U(X_n) \downarrow U(X)$ and $U(X) = \inf_{\mathbb{P}\in\mathcal{P}}
\crl{ \mathbb{E}^\mathbb{P} u\brak{X+ \sum_{t=1}^T \vartheta^*_t \Delta S_t} + \alpha(\p)}$.
\end{proof}

\noindent
{\bf Proof of Theorem \ref{thm2}}\\
Assume a medial limit exists, $u$ satisfies (U1)--(U3) and ${\cal P}$ fulfills (NA).
Then an application of Lemma \ref{lem:existence.strategy} with $X_n = X$ yields that
the supremum in \eqref{U1} is attained for every Borel measurable function 
$X \colon \Omega \to \mathbb{R}$ satisfying $U(X) \in \mathbb{R}$.

If in addition, (A1)--(A2) hold, we know from the proof of Theorem \ref{thm1} that 
$\phi(X) = -U(-X)$ is an increasing convex mapping from $B_Z$ to $\mathbb{R}$ satisfying 
condition (R2) of Theorem \ref{thm:dual} and 
$$
\phi^*_{C_Z}(q\q) = U^*_{C_Z}(q\mathbb{Q})= \left\{
\begin{array}{ll}
D_v(q\mathbb{Q},\mathcal{P}) & \mbox{ if } q = 0 \mbox{ or } \q \in {\cal Q}_Z\\
 \infty & \mbox{ if } q > 0 \mbox{ and } \q \in {\cal M}_Z \setminus {\cal Q}_Z. 
\end{array} \right.
$$
Moreover, by Lemma \ref{lem:existence.strategy}, $\phi$ fulfills (R3).
Hence, it follows from Theorem \ref{thm:dual} that 
\[
U(X) = -\phi(-X) = \inf_{(q, \mathbb{Q}) \in \mathbb{R}_+ \times {\cal M}_Z} 
\crl{q\mathbb{E}^{\q} X + U^*_{C_Z}(q\mathbb{Q})}
= \inf_{(q, \mathbb{Q}) \in \hat{\cal Q}_Z} 
\crl{q\mathbb{E}^{\q} X + U^*_{C_Z}(q\mathbb{Q})}
\]
for all $X\in B_Z$. \qed

\bigskip
\noindent 
{\bf Proof of Corollary \ref{cor}}\\
Note that 
\[W(X) = - \frac{1}{\lambda} \log(-U(X))\] for
\be \label{UW}
U(X) = \sup_{\vartheta \in \Theta} \inf_{\p \in {\cal P}} \mathbb{E}^{\p} u \brak{X + \sum_{t=1}^T \vartheta_t \Delta S_t}
\quad \mbox{and} \quad u(x) = - \exp(-\lambda x).
\ee
Clearly, $u$ satisfies (U1)--(U3), and under the assumptions of the corollary, ${\cal P}$ together with the trivial 
function $\alpha \equiv 0$ fulfill (A1)--(A2). Therefore, it follows from Theorem \ref{thm2} that
the supremum in \eqref{UW} is attained for all $X \in B_Z$. In particular, $U(X) \in (-\infty,0)$, and
therefore, $W(X) \in \mathbb{R}$ for all $X \in B_Z$. Furthermore,
\[
v(y) = \sup_{x \in \mathbb{R}} \crl{u(x) - xy} = \frac{y}{\lambda} \brak{\log \frac{y}{\lambda} - 1},
\]
from which it follows that
\[
\inf_{\p \in {\cal P}} D_v(q\q \,\| \, \p) = \frac{q}{\lambda}H(\q \,\| \, {\cal P}) +
  \frac{q}{\lambda} \brak{\log \frac{q}{\lambda}- 1}.\]
So, by Theorem \ref{thm2},
  \begin{align*} 
  W(X)&=-\frac{1}{\lambda}\log(-U(X))
  =-\frac{1}{\lambda}\log\Big(- \inf_{(q, \mathbb{Q}) \in \hat{\cal Q}_Z} 
  \brak{q \mathbb{E}^{\mathbb{Q}}X + \inf_{\p \in {\cal P}} D_v(q \q \,\| \, \p)}\Big)\\
  &=-\frac{1}{\lambda}\log\Big(- \inf_{q\in \mathbb{R}_+} \Big(q
  \Big(\inf_{\mathbb{Q} \in \mathcal{Q}_Z}(\mathbb{E}^{\mathbb{Q}}X + \frac{1}{\lambda} 
  H(\q \,\| \, {\cal P}))\Big) +\frac{q}{\lambda} \brak{\log \frac{q}{\lambda}-1}\Big)\Big).
  \end{align*}
Solving for the minimizing $q$ gives
$W(X)=\inf_{\mathbb{Q} \in \mathcal{Q}_Z} \crl{\mathbb{E}^{\mathbb{Q}}X + \frac{1}{\lambda} H(\q \,\| \, {\cal P})}$.
\qed

\setcounter{equation}{0}
\begin{appendix}
\section{Appendix}

\subsection{Properties of Example \ref{ex:mo}}
\label{ap:mo}

Clearly, ${\cal P}$ is a convex subset of ${\cal M}_Z$. So to prove that it satisfies (A1) for 
$\alpha \equiv 0$, it is enough to show that it is $\sigma({\cal M}_Z, C_Z)$-closed. To do that, let
$(\p_n)$ be a sequence in ${\cal P}$ converging in $\sigma({\cal M}_Z, C_Z)$ to a 
Borel probability measure $\p$. Then,
\[
\mathbb{E}^\mathbb{P}[S_t^{c^i} \wedge m] 
= \lim_n \mathbb{E}^{\p_n}[S_t^{c^i} \wedge m] \le C^i_t \quad \mbox{and} \quad 
\mathbb{E}^\mathbb{P}[S_t^{d^i} \wedge m] = \lim_n \mathbb{E}^{\p_n}[S_t^{d^i} \wedge m] \le D^i_t
\]
for all $t = 1, \dots, T$, $i = 1, \dots, I$ and $m \in \mathbb{N}$, from which it follows by monotone convergence that
\[
\mathbb{E}^\mathbb{P}[S_t^{c^i}] \le C^i_t \quad \mbox{and} \quad 
\mathbb{E}^\mathbb{P}[S_t^{d^i}] \le D^i_t \quad \mbox{for all $t$ and $i$.}
\]
Hence, ${\cal P}$ is $\sigma({\cal M}_Z, C_Z)$-closed.

Moreover, if $u \colon \Omega\times \mathbb{R} \to \mathbb{R}$ is a random 
utility function satisfying (U1)--(U3) and there exists a constant 
\[
q < p := \max_{1 \le i \le I} |c^i| \wedge \max_{1 \le i \le I} |d^i|
\] 
such that $u(\omega, x)/(1+ |x|^q)$ is bounded, then (A2) holds for $\alpha \equiv 0$ and $\beta(x) = x^{p/q}$.

  Now, let us assume that $s^{c^i}_0 < C^i_t$ and $s^{d^i}_0 < D^i_t$ for all $t$ and $i$.
  To show that ${\cal P}$ satisfies (NA), we assume for notational simplicity that $T=2$ and $a_t=b_t=1$ for $t=1,2$. 
  Then $\Omega$ can be identified with $(0,\infty) \times (0,\infty)$.
  The general case follows from similar arguments. 
  Choose a $\p \in {\cal P}$ and disintegrate it as $\mathbb{P}=\p_1 \otimes K$, where $\p_1$ is the 
  first marginal distribution (corresponding to the distribution of $S_1$) and $K$ is a transition probability kernel
  (corresponding to the conditional distribution of $S_2$ given $S_1$).
  For every $\varepsilon \in (0, 1)$, denote by $\p_1^\varepsilon$ and $K^\varepsilon$ 
  the measure and kernel given by 
  \[\p_1^\varepsilon:= \frac{\delta_{(1-\varepsilon)s_0} +\delta_{(1+ \varepsilon)s_0}}{2}
  \quad \text{and} \quad
  K^\varepsilon_x := \frac{\delta_{(1- \varepsilon)x}+\delta_{(1+\varepsilon)x}}{2}.\]
  Then, the measure
  \[ \mathbb{P}^\varepsilon:=(\varepsilon \p_1 +(1-\varepsilon)\p_1^\varepsilon)\otimes
  (\varepsilon K+(1-\varepsilon)K^\varepsilon)\]
  dominates $\p$ and does not admit arbitrage. It remains to show that 
  $\mathbb{P}^\varepsilon$ belongs to $\mathcal{P}$ for some $\varepsilon >0$. 
  First, note that for $m = c^i$ or $d^i$,
  \[\mathbb{E}^{\mathbb{P}^\varepsilon}S_1^m
  =\varepsilon\mathbb{E}^{\p}S_1^m + (1-\varepsilon) \frac{(1-\varepsilon)^m +(1+\varepsilon)^m}{2} s^m_0
  \to s_0^m \quad \mbox{as } \varepsilon \to 0.\]
  This shows that the moment conditions for $S_1$ under $\p^{\varepsilon}$ 
  are satisfied as soon as $\varepsilon > 0$ is sufficiently small.
  Moreover, for $m = c^i$ or $d^i$, one has
  \be \label{exp}
  \mathbb{E}^{\mathbb{P}^\varepsilon}S_2^m
  = \varepsilon^2 \mathbb{E}^{\mathbb{P}} S_2^m
  + \varepsilon(1-\varepsilon) \mathbb{E}^{\p_1\otimes K^\varepsilon} S_2^m
  + \varepsilon(1-\varepsilon) \mathbb{E}^{\p_1^\varepsilon \otimes K} S_2^m
  +(1-\varepsilon)^2 \mathbb{E}^{\p_1^\varepsilon\otimes K^\varepsilon} S_2^m. \ee
  The term $\mathbb{E}^{\p_1^\varepsilon\otimes K^\varepsilon} S_2^m$ converges to $s_0^m$ for $\varepsilon\to 0$. 
  So if we can show that the other expectations in \eqref{exp} are bounded in $\varepsilon$, 
  it follows that $S_2$ satisfies the moment constraints for $\varepsilon > 0$ small enough.
  The first expectation $\mathbb{E}^{\mathbb{P}} S_2^{m}$ is independent of $\varepsilon$ and finite 
  since $\p$ belongs to ${\cal P}$. The second expectation satisfies
  \[\mathbb{E}^{\p_1 \otimes K^\varepsilon} S_2^m
  = \frac{\mathbb{E}^{\p_1} (1-\varepsilon)^m S_1^m + \mathbb{E}^{\p_1} (1+\varepsilon)^m S_1^m}{2}
  \to \mathbb{E}^{\p_1} S_1^m
  \le C^i_1 \quad \mbox{for } \varepsilon \to 0. \]
  Finally, note that one can change $K$ on a $\p_1$-zero set and still have $\mathbb{P}=\p_1 \otimes K$.
  Therefore, one can assume that $K_{(1 \pm \varepsilon_n) s_0} = \delta_{s_0}$ for a 
  sequence of positive numbers $(\varepsilon_n)$ converging to $0$. 
  Then $\mathbb{E}^{\p_1^{\varepsilon_n} \otimes K} S_2^m = s_0^m$.
 Hence, the moment conditions for $S_2$ under $\p^{\varepsilon_n}$ 
 hold too for $\varepsilon_n$ close enough to $0$, showing that ${\cal P}$ fulfills (NA).
 
\subsection{Properties of Example \ref{ex:wball}}
\label{ap:wball}

Obviously, $W(\cdot,\p^*)^p$, and consequently also ${\cal P}$, are convex. Moreover, by Kantorovich duality
(see e.g. Theorem 5.10 in \cite{Vi}), one has
\[
W_p(\p, \p^*)^p = \sup_{f \in C_b(\Omega)\, ,\, g \in C_b(\Omega) \, , \,
f + g \le d^p} \brak{\mathbb{E}^{\p} f + \mathbb{E}^{\p^*} g},
\]
from which it is easy to see that ${\cal P}$ is $\sigma({\cal M}_Z, {\cal P}_Z)$-closed.
This shows that (A1) holds for $\alpha \equiv 0$.

Next, note that $Z(\omega) = s_0 + T + e^{\rho T} T^{1-1/\kappa} d(\omega, \omega^*)$ defines a continuous 
function $Z \colon \Omega\to [1, \infty)$ with compact sublevel sets $\crl{Z \le z}$, $z \in \mathbb{R}_+$, such 
that 
\[
Z(\omega) \ge s_0 + T + \sum_{t=1}^T |\varphi(\omega_{t,2}) - \varphi(\omega^*_{t,2})|
\ge 1 \vee \brak{s_0 + \sum_{t=1}^T \omega_{t,2}} = 1 \vee \sum_{t=0}^T |S_t|.
\]
Since there exists a constant $c \in \mathbb{R}_+$ such that $Z(\omega)^p \leq c(1+ d(\omega,\omega^*)^p)$, 
$W_p$ satisfies the triangle inequality (see e.g. Chapter 6 of \cite{Vi}) and
$\mathbb{E}^{\p} d(\cdot,\omega^*)^p = W_p(\p,\delta_{\omega^*})^p$, one has
\[ 
\mathbb{E}^{\p} Z^p \le  c( 1+ W_p(\p,\delta_{\omega^*})^p ) 
\leq 2^{p-1} c( 1+ W_p(\p,\p^*)^p + W_p(\p^*, \delta_{\omega^*})^p )
\leq 2^{p-1} c ( 1+ W_p(\p,\p^*)^p + \mathbb{E}^{\p^*} Z^p ) \] for all $\p \in {\cal M}_Z$. 
In particular, if $u(\omega,x)/(1+ |x|^q)$ is bounded for a constant $q < p$, then 
\[
\mathbb{E}^{\p} u(-\beta(Z)) \geq -c\brak{1+ W_p(\p,\p^*)^p}
\] 
for $\beta(x) = x^{p/q}$, a new constant $c \in \mathbb{R}_+$ and all $\p \in {\cal P}$, 
showing that (A2) holds for $\alpha \equiv 0$.

  To prove that ${\cal P}$ satisfies (NA), we again assume $T=2$ and $a_t=b_t =1$ for $t =1,2$.
  The general case follows analogously.
  Choose a $\mathbb{P}\in\mathcal{P}$ and disintegrate it as $\mathbb{P}=\mathbb{P}_1\otimes K$.
  Similarly, write $\p^* = \p^*_1 \otimes K^*$, and define for $\lambda \in (0,1)$, 
  \[
  \p^{\lambda}_1 := \lambda \frac{\delta_{(1-\lambda) s_0} + \delta_{(1+\lambda)s_0}}{2} 
  + (1-\lambda) \p^*_1 \quad \mbox{and} \quad K^{\lambda}_x :=   
  \lambda \frac{\delta_{(1-\lambda)x} + \delta_{(1+\lambda)x}}{2} + (1-\lambda) K^*_x.
  \]
  Then the measure $\p^{\lambda} := \p^{\lambda}_1 \otimes K^{\lambda}$ does not admit arbitrage.
  Moreover, there exists a $\lambda \in (0,1)$ such that $\mathbb{E}^{\p^{\lambda}} d(\cdot,\omega^*)^p < \infty$ 
  and $W_p(\p^{\lambda}, \p^*) \le \eta/2$. Now choose a measure $\tilde{\mathbb{P}}_1$ equivalent to $\mathbb{P}_1$ and
  a transition probability kernel $\tilde{K}$ such that for all $x > 0$, $\tilde{K}_x$ is equivalent to $K_x$
  and the three expectations $\mathbb{E}^{\tilde{\mathbb{P}}_1 \otimes \tilde{K}}d(\cdot,\omega^*)^p$,
  $\mathbb{E}^{\tilde{\mathbb{P}}_1\otimes K^{\lambda}}d(\cdot,\omega^*)^p$ and 
  $\mathbb{E}^{\p_1^{\lambda} \otimes \tilde{K}}d(\cdot,\omega^*)^p$ are finite.
  For each $\varepsilon \in (0,1)$, the measure
  \[\mathbb{P}^{\varepsilon, \lambda} := (\varepsilon \tilde{\mathbb{P}}_1 + (1-\varepsilon) \p_1^{\lambda})
  \otimes (\varepsilon \tilde{K} + (1-\varepsilon) K^{\lambda})\]
  does not admit arbitrage and $\mathbb{P} \ll \tilde{\mathbb{P}}_1 \otimes \tilde{K} \ll \mathbb{P}^{\varepsilon, \lambda}$.
  Since $W_p(\cdot,\p^*)^p$ is convex, $W_p(\mathbb{P}^{\varepsilon, \lambda},\p^*)^p$ is dominated by
  \[
  \varepsilon^2 W_p(\tilde{\mathbb{P}}_1\otimes \tilde{K},\p^*)^p
  +\varepsilon(1-\varepsilon) W_p(\tilde{\mathbb{P}}_1\otimes K^{\lambda},\p^*)^p
  +(1-\varepsilon) \varepsilon W_p(\p^{\lambda}_1\otimes \tilde{K},\p^*)^p
  + (1-\varepsilon)^2 W_p(\p^{\lambda},\p^*)^p.
  \]
  Due to $\mathbb{E}^{\p^*}d(\cdot,\omega^*)^p<\infty$, one obtains from the triangle inequality that
  \[W_p(\tilde{\mathbb{P}}_1\otimes \tilde{K},\p^*)
  \leq W_p(\tilde{\mathbb{P}}_1 \otimes \tilde{K},\delta_{\omega^*}) + W_p(\delta_{\omega^*},\p^*)
  = \brak{\mathbb{E}^{\tilde{\mathbb{P}}_1\otimes \tilde{K}} d(\cdot,\omega^*)^p}^{1/p} 
  + \brak{\mathbb{E}^{\p^*} d(\cdot,\omega^*))^p}^{1/p} <\infty,\]
  and similarly, $W_p(\tilde{\mathbb{P}}_1\otimes K^{\lambda},\p^*) < \infty$ as well as
  $W_p(\p^{\lambda}_1\otimes \tilde{K},\p^*)<\infty$.
  This shows that $W_p(\p^{\varepsilon, \lambda}, \p^*) \le \eta$ for $\varepsilon > 0$ small enough,
  proving that ${\cal P}$ satisfies (NA).
  
\subsection{Properties of Example \ref{ex:wpen}}
\label{ap:wpen}

It is easy to see that ${\cal P}$ and $\alpha$ are convex. Moreover, it follows from the arguments in Appendix \ref{ap:wball} 
that all sublevel sets ${\cal P}_c$, $c \in \mathbb{R}_+$, are $\sigma({\cal M}_Z, C_Z)$-closed, and for 
each $c > 0$, ${\cal P}_c$ satisfies (NA). So (A1) holds and ${\cal P}$ fulfills (NA).
Finally, if $u(\omega,x)/(1+ |x|^q)$ is bounded for a constant $q < p$, one obtains as in 
Appendix \ref{ap:wball} that
\be \label{Eu}
\mathbb{E}^{\p} u(-\beta(Z)) \geq -c \brak{1+ W_{(p+q)/2} (\p,\p^*)^{(p+q)/2}}
\ge -c \brak{1+ W_p(\p,\p^*)^{(p+q)/2}} \ee
for $\beta(x) = x^{(p+q)/2q}$, a constant $c \in \mathbb{R}_+$ and all $\p \in {\cal P}$. This shows that
\[
\inf_{\mathbb{P} \in \mathcal{P}} \crl{\mathbb{E}^\mathbb{P}u(-\beta(Z)) + \eta W_p(\p,\p^*)^p} > - \infty,
\]
and (A2) holds.
\end{appendix}

\bibliographystyle{plain}

\begin{thebibliography}{10}

\bibitem{backhoff2015robust}
J.~Backhoff and J.~Fontbona (2016).
\newblock {\em Robust utility maximization without model compactness.}
\newblock SIAM J. Fin. Math. 7(1), 70--103. 

\bibitem{B}
D.~Bartl (2018).
\newblock {\em Exponential utility maximization under model uncertainty for unbounded endowments.}
\newblock Forthcoming in Ann. Appl. Prob.

\bibitem{BS}
D.~P.~Bertsekas and S.~E.~Shreve (1978).
\newblock Stochastic Optimal Control: The Discrete Time Case.
\newblock Academic Press New York.

\bibitem{BC}
R.~Blanchard and L.~Carassus (2018).
\newblock {\em Robust optimal investment in discrete time for unbounded utility function.}
\newblock Ann. Appl. Proba. 28(3), 1856-1892.

\bibitem{BMS}
G.~Bordigoni, A.~Matoussi and M. Schweizer (2007).
\newblock 
{\em A stochastic control approach to a robust utility maximization problem.}
\newblock
in Stochastic Analysis and Applications, Abel Symposium 2, 125--151.

\bibitem{BN}
B.~Bouchard, M.~Nutz (2015).
\newblock {\em Arbitrage and duality in nondominated discrete-time models.}
\newblock Ann. Appl. Proba. 25(2), 823--859.

\bibitem{CDK}
P.~Cheridito, F.~Delbaen and M.~Kupper (2006).
\newblock {\em Dynamic monetary risk measures for bounded discrete-time processes.}
\newblock Electr. J. Probab. 11, 57--106.

\bibitem{CK}
P.~Cheridito and M.~Kupper (2006).
\newblock {\em Composition of time-consistent dynamic monetary risk measures in discrete time.}
\newblock Int. J. Theor. Appl. Fin. 14(1), 137--162.

\bibitem{Ch}
G.~Choquet (1959).
\newblock {\em Forme abstraite du th{\'e}or{\`e}me de capacitabilit{\'e}.}
\newblock Annales de L'Institut Fourier 9, 83--89.

\bibitem{KF}
K.~Fan (1953).
\newblock {\em Minimax theorems.}
\newblock Proc. Nat. Academy of Sciences 39(1), 42--47.
  
\bibitem{FS}
H.~F\"ollmer and A.~Schied (2004). 
\newblock Stochastic Finance. An Introduction in Discrete Time.
\newblock Walter de Gruyter GmbH \& Co., Berlin, New York.

\bibitem{gundel2005utility}
A.~Gundel (2005). 
\newblock {\em Robust utility maximization for complete and incomplete market models.}
\newblock Finance and Stochastics 9(2), 151--176.

\bibitem{HS}
L.P.~Hansen and T.J.~Sargent (2001) 
\newblock {\em Robust control and model uncertainty.}
\newblock Am. Econ. Rev. 91, 60--66.

\bibitem{jacod1998local}
J.~Jacod and A.~N.~Shiryaev (1998).
\newblock {\em Local martingales and the fundamental asset pricing theorems in the
  discrete-time case.}
\newblock Finance and Stochastics 2(3), 259--273.

\bibitem{MPZ}
A.~Matoussi, D.~Possama\"i and C.~Zhou (2015).
\newblock
{\em Robust utility maximization in nondominated models with 2BSDE: the uncertain volatility model.}
\newblock 
Math. Fin. 25(2), 258--287.
  
\bibitem{Meyer}
P.~A.~Meyer (1973).
\newblock {\em Limites m\'ediales d'apr\`es Mokobozki.}
\newblock S\'eminaire de Probabilit\'es 7, 198--204.

\bibitem{NN}
A.~Neufeld and M.~Nutz (2018).
\newblock
{\em Robust utility maximization with L\'evy processes.}
\newblock Math. Fin. 28 (1), 82--105.

\bibitem{NS}
A.~Neufeld and M.~\v{S}iki\'{c} (2018).
\newblock {\em Robust utility maximization in discrete-time markets with friction.}
\newblock SIAM J. Contr. Optim. 56(3), 1912--1937.

\bibitem{Normann}
D.~Normann (1976).
\newblock {\em Martin's axiom and medial functions.}
\newblock Math. Scand. 38, 167--176.

\bibitem{N14}
M.~Nutz (2014).
\newblock {\em Superreplication under model uncertainty in discrete time.}
\newblock Fin. Stoch. 18(4), 791--803.

\bibitem{N16}
M.~Nutz (2016).
\newblock {\em Utility maximization under model uncertainty in discrete time.}
\newblock Math. Fin. 26(2), 252--268.

\bibitem{owari2011robust}
K.~Owari (2011).
\newblock {\em Robust utility maximization with unbounded random endowment.}
\newblock Adv. Math. Econ. 14, 147--181.

\bibitem{quenez2004optimal}
M.~C.~Quenez (2004).
\newblock {\em Optimal portfolio in a multiple-priors model.}
\newblock Sem. Stoch. Analysis, Random Fields and Appl. IV. Birkh{\"a}user, Basel.  

\bibitem{S05}
A.~Schied (2005).
\newblock 
{\em Optimal investments for robust utility functionals in complete market models.}
\newblock Math. Oper. Res. 30(2), 750--764.

\bibitem{S07}
A.~Schied (2007).
\newblock {\em Optimal investments for risk- and ambiguity-averse preferences: a duality approach.}
\newblock Fin. Stoch. 11, 107--129.

\bibitem{HT} H. Tong (1952). 
\newblock {\em Some characterizations of normal and perfectly normal spaces.}
\newblock Duke Math. J. 19, 289--292.

\bibitem{Vi} C. Villani (2009).
Optimal Transport: Old and New. Grundlehren der Mathematischen 338. Springer.
Wissenschaften. 

\end{thebibliography}

\end{document}